\documentclass[12pt]{article}
\usepackage{amsmath,amssymb}
\usepackage{amsmath,amsthm,mathtools}
\usepackage{graphicx,epstopdf}
\usepackage{algorithm}
\usepackage{algorithmic}
\usepackage{float}
\usepackage{hyperref}       

\DeclareMathOperator*{\argmin}{arg\,min}

\usepackage{color}
\usepackage{amsfonts,amssymb}       

\usepackage{authblk}

\usepackage[authoryear]{natbib}

\newtheorem{Th}{Theorem} 

\newtheorem{prop}{Proposition}

\theoremstyle{definition}

\theoremstyle{remark}

\usepackage{geometry}
\usepackage{pdflscape}
\usepackage{array}
\title{Why do similarity matching objectives lead to Hebbian/anti-Hebbian networks?}

\author[1]{Cengiz Pehlevan} 
\author[1,2]{Anirvan M. Sengupta} 
\author[1,3]{Dmitri B. Chklovskii} 
\affil[1]{Center for Computational Biology, Flatiron Institute, New York, NY}
\affil[2]{Physics and Astronomy Department, Rutgers University, New Brunswick, NJ}
\affil[3]{NYU Langone Medical Center, New York, NY}

 \date{\vspace{-5ex}}
 
\begin{document}

\maketitle

\begin{abstract}

Modeling self-organization of neural networks for unsupervised learning using Hebbian and anti-Hebbian plasticity has a long history in neuroscience. Yet, derivations of single-layer networks with such local learning rules from principled optimization objectives became possible only recently, with the introduction of similarity matching objectives. What explains the success of similarity matching objectives in deriving neural networks with local learning rules? Here, using dimensionality reduction as an example, we introduce several variable substitutions that illuminate the success of similarity matching. We show that the full network objective may be optimized separately for each synapse using local learning rules both in the offline and online settings. We formalize the long-standing intuition of the rivalry between Hebbian and anti-Hebbian rules by formulating a min-max optimization problem. We introduce a novel dimensionality reduction objective using fractional matrix exponents. To illustrate the generality of our approach, we apply it to a novel formulation of dimensionality reduction combined with whitening. We confirm numerically that the networks with learning rules derived from principled objectives perform better than those with heuristic learning rules. 
\end{abstract}

\section{Introduction}

The human brain generates complex behaviors via the dynamics of electrical activity in a network of $\sim10^{11}$ neurons each making $\sim10^{4}$ synaptic connections. As there is no known centralized authority determining which specific connections a neuron makes or specifying the weights of individual synapses, synaptic connections must be established based on local rules.  Therefore, a major challenge in neuroscience is to determine local synaptic learning rules that would ensure that the network acts coherently, i.e. guarantee robust network self-organization.

Much work has been devoted to the self-organization of neural networks for solving unsupervised computational tasks using Hebbian and anti-Hebbian learning rules \citep{foldiak1990forming,foldiak1989adaptive,rubner1989self,rubner1990development,carlson1990anti,plumbley1993hebbian,leen1991,plumbley1993efficient,linsker1997local}. 
Unsupervised setting is natural in biology because large-scale labeled datasets are typically unavailable. Hebbian and anti-Hebbian learning rules are biologically plausible 
because they are local: The weight of an (anti-)Hebbian synapse is proportional to the (minus) correlation in activity between the two neurons the synapse connects. 

In networks for dimensionality reduction, for example, feedforward connections use Hebbian rules and lateral - anti-Hebbian, Figure 1. Hebbian rules attempt to align each neuronal feature vector, whose components are the weights of synapses impinging onto the neuron, with the input space direction of greatest variance. Anti-Hebbian rules mediate competition among neurons which prevents their feature vectors from aligning in the same direction. A rivalry between the two kinds of rules results in the equilibrium where synaptic weight vectors span the principal subspace of the input covariance matrix, i. e. the subspace spanned by the eigenvectors corresponding to the largest eigenvalues.

However, in most existing single-layer networks, Figure 1, Hebbian and anti-Hebbian learning rules were postulated rather than derived from a principled objective. Having such derivation should yield better performing rules and deeper understanding than has been achived using heuristic rules. But, until recently, all derivations of single-layer networks from principled objectives led to biologically implausible non-local learning rules, where the weight of a synapse depends on the activities of neurons other than the two the synapse connects. 

Recently, single-layer networks with local learning rules have been derived from similarity matching objective functions \citep{pehlevan2015MDS,pehlevan2014NMF,hu2014SMF}. But why do similarity matching objectives lead to neural networks with local, Hebbian and anti-Hebbian learning rules? A clear answer to this question has been lacking.

Here, we answer this question by performing several illuminating variable transformations. Specifically, we reduce the full network optimization problem to a set of trivial optimization problems for each synapse which can be solved locally. Eliminating neural activity variables leads to a min-max objective in terms of feedforward and lateral synaptic weight matrices. This finally formalizes the long-held intuition about the adversarial relationship of Hebbian and anti-Hebbian learning rules.

In this paper, we make the following contributions. In Section \ref{S2}, we present a more transparent derivation of the previously proposed online similarity matching algorithm for Principal Subspace Projection (PSP). In Section \ref{S3}, we propose a novel objective for PSP combined with spherizing, or whitening, the data, which we name Principal Subspace Whitening (PSW), and derive from it a biologically plausible online algorithm. Also, in Sections \ref{S2} and \ref{S3}, we demonstrate that stability in the offline setting guarantees projection onto the principal subspace and give principled learning rate recommendations. In Section \ref{S4}, by eliminating activity variables from the objectives, we derive min-max formulations of PSP and PSW which yield themselves to game-theoretical interpretations. In Section \ref{S5}, by expressing the optimization objectives in terms of feedforward synaptic weights only, we arrive at novel formulations of dimensionality reduction in terms of fractional
powers of matrices. In Section \ref{S6}, we demonstrate numerically that the performance of our online algorithms is superior to the heuristic ones.

\section{From similarity matching to Hebbian/anti-Hebbian networks for PSP}\label{S2}

\subsection{Derivation of a mixed PSP from similarity matching}

The PSP problem is formulated as follows. Given $T$ centered input data samples, ${\bf x}_t \in \mathbb{R}^n$, find $T$ projections, ${\bf y}_t \in \mathbb{R}^k$, onto the principal subspace ($k\le n$), i.e. the subspace spanned by eigenvectors corresponding to the $k$ top eigenvalues of the input covariance matrix:
\begin{align}\label{Cdef}
{\bf C} \equiv  \frac 1 T \sum_{t=1}^{T} {\bf x}_t {\bf x}_t^\top = \frac 1 T {\bf X}{\bf X}^\top,
\end{align}
where we resort to a matrix notation by concatenating input column vectors into ${\bf X}=\left[{\bf x}_1,\ldots,{\bf x}_T\right]$. Similarly, outputs are ${\bf Y}=\left[{\bf y}_1,\ldots,{\bf y}_T\right]$.

Our goal is to derive a biologically plausible single-layer neural network implementing PSP by optimizing a principled objective. Biological plausibility requires that the learning rules are local, i.e. synaptic weight update depends on the activity of only the two neurons the synapse connects. The only PSP objective known to yield a single-layer neural network with local learning rules is based on similarity matching \citep{pehlevan2015MDS}. This objective, borrowed from Multi-Dimensional Scaling (MDS), minimizes the mismatch between the similarity of inputs and outputs \citep{mardia1980multivariate,williams01ona,cox2000multidimensional}:
\begin{align}\label{SM}
{\rm PSP:} \hspace{3cm} \min_{{\bf Y}\in \mathbb{R}^{k\times T}} \frac 1{T^2} \left\Vert {\bf X}^\top{\bf X}-{\bf Y}^\top{\bf Y}\right\Vert_F^2 .
\end{align}
Here, similarity is quantified by the inner products between all pairs of inputs (outputs) comprising the Grammians ${\bf X}^\top{\bf X}$ (${\bf Y}^\top{\bf Y}$). 

One can understand intuitively that the objective \eqref{SM} is optimized by the projection onto the principal subspace by considering the following (for a rigorous proof see \citep{pehlevan2015normative,mardia1980multivariate,cox2000multidimensional}). First, substitute a Singular Value Decomposition (SVD) for matrices ${\bf X}$ and ${\bf Y}$ and note that the mismatch is minimized by matching right singular vectors of ${\bf Y}$ to that of ${\bf X}$. Then, rotating the Grammians to the diagonal basis reduces the minimization problem to minimizing the mismatch between the corresponding singular values squared. Therefore, ${\bf Y}$ is given by the top $k$ right singular vectors of ${\bf X}$ scaled by corresponding singular values. As the objective \eqref{SM} is invariant to the left-multiplication of  ${\bf Y}$ by an orthogonal matrix, it has infinitely many degenerate solutions. One such solution corresponds to the Principal Component Analysis (PCA). 

Unlike non-neural-network formulations of PSP or PCA, similarity matching outputs principal components (scores) rather than principal eigenvectors of the input covariance (loadings). Such difference in formulation is motivated by our interest in PSP or PCA neural networks \citep{diamantaras1996principal} that output principal components, ${\bf y}_t$, rather than principal eigenvectors. Principal eigenvectors are not transmitted downstream of the network but can be recovered computationally from the synaptic weight matrices. Although synaptic weights do not enter the objective \eqref{SM}, in previous work \citep{pehlevan2015MDS}, they arose naturally in the derivation of the online algorithm (see below) and stored correlations between input and output neural activities.

Next, we derive the min-max PSP objective from Eq. \eqref{SM}, starting with expanding the square of the Frobenius norm:
\begin{align}\label{SMfull}
 \argmin_{{\bf Y}\in \mathbb{R}^{k\times T}} \frac 1{T^2} \left\Vert {\bf X}^\top{\bf X}-{\bf Y}^\top{\bf Y}\right\Vert_F^2 =  \argmin_{{\bf Y}\in \mathbb{R}^{k\times T}} \frac{1}{T^2} {\rm Tr}\left(-2{\bf X}^\top{\bf X}{\bf Y}^\top{\bf Y} + {\bf Y}^\top{\bf Y}{\bf Y}^\top{\bf Y} \right).
\end{align}
We can rewrite Eq. \eqref{SMfull} by introducing two new dynamical variable matrices in place of covariance matrices $\frac 1T {\bf X}{\bf Y}^\top$ and $\frac 1T {\bf Y}{\bf Y}^\top$:
\begin{align}\label{S0}
\min_{{\bf Y}\in \mathbb{R}^{k\times T}} \min_{{\bf W}\in \mathbb{R}^{k\times n}}\max_{{\bf M}\in \mathbb{R}^{k\times k}} \, &L_{PSP}({\bf W},{\bf M},{\bf Y}),  {\rm \; \; where} 
\end{align}
\begin{align}\label{SMMW}
\qquad L_{PSP}({\bf W},{\bf M},{\bf Y}) &\equiv {\rm Tr}\left( -\frac{4}{T}{\bf X}^\top{\bf W}^\top{\bf Y} + \frac{2}{T} {\bf Y}^\top{\bf M}{\bf Y} \right) + 2 {\rm Tr}\left({\bf W}^\top{\bf W}\right) -  {\rm Tr}\left({\bf M}^\top{\bf M}\right).
\end{align}  
To see that Eq. \eqref{SMMW} is equivalent to Eq. \eqref{SMfull} find optimal ${\bf W}^*=\frac 1T {\bf Y}{\bf X}^\top$ and ${\bf M}^*=\frac 1T {\bf Y}{\bf Y}^\top$ by setting the corresponding derivatives of objective \eqref{SMMW} to zero. Then, substitute ${\bf W}^*$ and ${\bf M}^*$ into Eq. \eqref{SMMW} to obtain \eqref{SMfull}. 

Finally, we exchange the order of minimization with respect to ${\bf Y}$ and ${\bf W}$ as well as the order of minimization with respect to ${\bf Y}$ and maximization with respect to ${\bf M}$ in Eq. \eqref{SMMW}. The last exchange is justified by the saddle point property (see Proposition \ref{minmaxPSP} in Appendix \ref{A1}). Then, we arrive at the following min-max optimization problem:
\begin{align}\label{SMMW2}
\min_{{\bf W}\in \mathbb{R}^{k\times n}}\max_{{\bf M}\in \mathbb{R}^{k\times k}} \min_{{\bf Y}\in \mathbb{R}^{k\times T}}L_{PSP}({\bf W},{\bf M},{\bf Y}) ,
\end{align}  
where $L_{PSP}({\bf W},{\bf M},{\bf Y})$ is defined in Eq. \eqref{SMMW}. We call this a mixed objective because it includes both output variables, ${\bf Y}$, and covariances, ${\bf W}$ and ${\bf M}$.

\subsection{Offline PSP algorithm}

In this section, we present an offline optimization algorithm to solve the PSP problem and analyze fixed points of the corresponding dynamics. These results will be used in the next Section for the biologically plausible online algorithm implemented by neural networks.

In the offline setting, we can solve Eq. \eqref{SMMW2} by the alternating optimization approach used commonly in neural networks literature \citep{olshausen1996emergence,olshausen1997sparse,arora2015simple}. We, first, minimize with respect to ${\bf Y}$ while keeping ${\bf W}$ and ${\bf M}$ fixed,
\begin{align}
{\bf Y}^* = \argmin_{{\bf Y}\in \mathbb{R}^{k\times T}} L_{PSP}({\bf W},{\bf M},{\bf Y}),
\end{align}
and, second, make a gradient descent-ascent step with respect to ${\bf W}$ and ${\bf M}$ while keeping ${\bf Y}$ fixed:
%
\begin{align}\label{gda}
\left[\begin{array}{c} {\bf W} \hspace{1cm} {\bf M} \end{array}\right] &\longleftarrow \left[\begin{array}{c} {\bf W} \hspace{1cm} {\bf M} \end{array}\right] +  \left[\begin{array}{c} - \eta \frac {\partial L_{PSP}({\bf W},{\bf M},{\bf Y}^*)}{\partial {\bf W}} \hspace{1cm} \frac{\eta}{\tau}\frac {\partial L_{PSP}({\bf W},{\bf M},{\bf Y}^*)}{\partial {\bf M}} \end{array}\right],
\end{align}
where $\eta$ is the ${\bf W}$ learning rate and $\tau > 0$ is a ratio of learning rates for ${\bf W}$ and ${\bf M}$. In Appendix \ref{LSPSP}, we analyze how $\tau$ affects linear stability of the fixed point dynamics. These two phases are iterated until convergence (Algorithm 1)\footnote{This alternating optimization is identical to a gradient descent-ascent (see Proposition \ref{lem1} in Appendix \ref{chainrule}) in ${\bf W}$ and ${\bf M}$ on the objective:  
	\begin{align*}
	l_{PSP}({\bf W},{\bf M}) \equiv  \min_{{\bf Y}\in \mathbb{R}^{k\times T}} L_{PSP}({\bf W},{\bf M},{\bf Y}).
	\end{align*}
}.

\begin{algorithm}[H]
	\caption{Offline min-max PSP\label{offlinePSPalg}}
	\begin{algorithmic}[1]
		\STATE Initialize ${\bf W}$. Initalize ${\bf M}$ as a positive definite matrix.
		\STATE Iterate until convergence:
		\begin{ALC@g}
			\STATE
			Minimize Eq. \eqref{SMMW} with respect to ${\bf Y} $, keeping ${\bf W}$ and ${\bf M}$ fixed: 
			\begin{align}\label{Y}
			{\bf Y} = {\bf M}^{-1}{\bf W}{\bf X}.
			\end{align}

			\STATE  Perform a gradient descent-ascent step with respect to ${\bf W}$ and ${\bf M}$ for a fixed ${\bf Y}$:
			\begin{align}
			{\bf W}&\longleftarrow {\bf W} + 2\eta \left(\frac 1T {\bf Y}{\bf X}^\top- {\bf W}\right), \nonumber \\
			{\bf M} & \longleftarrow {\bf M}  +  \frac \eta{\tau} \left(\frac 1T {\bf Y}{\bf Y}^\top-{\bf M}\right).
			\end{align}
			where the step size, $0<\eta<1$, may depend on the iteration.
		\end{ALC@g}
		
	\end{algorithmic}
\end{algorithm}
Optimal ${\bf Y}$ in Eq. \eqref{Y} exists because ${\bf M}$ stays positive definite if initialized as such. 

\subsection{Linearly stable fixed points of Algorithm 1 correspond to the PSP}

Here we demonstrate that convergence of Algorithm \ref{offlinePSPalg} to fixed ${\bf W}$ and ${\bf M}$ implies that ${\bf Y}$ is a PSP of ${\bf X}$. To this end, we approximate the gradient descent-ascent dynamics in the limit of small learning rate with the system of  differential equations:
\begin{align}\label{gdPSP}
{\bf Y}(t) &= {\bf M}^{-1}(t){\bf W}(t){\bf X}, \nonumber \\
\frac{d {\bf W}(t)}{dt} &= \frac 2{T}{\bf Y}(t){\bf X}^\top -2{\bf W}(t), \nonumber \\
{\tau}\frac{d{\bf M}(t)}{dt} &= \frac{1}{T}{\bf Y}(t){\bf Y}(t)^\top - {\bf M}(t),
\end{align}
where $t$ is now the time index for gradient descent-ascent dynamics. 

To state our main result in Theorem \ref{mainLSPSP}, we define the ``filter matrix" ${\bf F}(t)$ whose rows are ``neural filters"
\begin{align}\label{fdef}
{\bf F}(t) := {\bf M}^{-1}(t){\bf W}(t),
\end{align}
so that, according to Eq. \eqref{Y}, 
\begin{align}
{\bf Y}(t) = {\bf F}(t){\bf X}.
\end{align}

\begin{Th}\label{mainLSPSP} Fixed points of the dynamical system \eqref{gdPSP} have the following properties: 
	\begin{enumerate}
		\item The neural filters, ${\bf F}$, are orthonormal, i.e. ${\bf F F^\top}={\bf I}$.
		\item The neural filters span a $k$-dimensional subspace in $\mathbb{R}^n$ spanned by some $k$ eigenvectors of the input covariance matrix.
		\item Stability of a fixed point requires that the neural filters span the {\bf principal} subspace of ${\bf X}$.
		\item Suppose the neural filters span the principal subspace. Define
		\begin{align}
		\gamma_{ij} := 2+\frac{\left(\sigma_i-\sigma_j\right)^2}{\sigma_i\sigma_j},
		\end{align}
		where $i = 1,\ldots,k$, $j=1,\ldots,k$ and $\lbrace \sigma_1,\ldots,\sigma_k\rbrace$ are the top $k$ principal eigenvalues of ${\bf C}$. We assume ${\sigma}_k \neq \sigma_{k+1}$. This fixed point is linearly stable if and only if:
		\begin{align}\label{tauPSP}
		\tau < \frac{1}{2-4/\gamma_{ij}}
		\end{align}
		for all $(i,j)$ pairs. By linearly stable we mean that linear perturbations of ${\bf W}$ and ${\bf M}$ converge to a configuration in which the new neural filters are merely rotations within the principal subspace of the original neural filters. 
	\end{enumerate}
\end{Th}
\begin{proof} See Appendix \ref{LSPSP}. \end{proof}

Based on Theorem \ref{mainLSPSP} we claim that, provided the dynamics converges to a fixed point, Algorithm \ref{offlinePSPalg} has found a PSP of input data. Note that the orthonormality of the neural filters is desired and consistent with PSP since, in this approach, outputs, ${\bf Y}$, are interpreted as coordinates with respect to a basis spanning the principal subspace. 

Theorem \ref{mainLSPSP} yields a practical recommendation for choosing learning rate parameters in simulations. In a typical situation, one will not know the eigenvalues of the covariance matrix a priori but can rely on the fact, $\gamma_{ij}\geq 2$. Then, Eq. \eqref{tauPSP} implies that for $\tau \leq 1/2$ the principal subspace is linearly stable leading to numerical convergence and stability.

\subsection{Online neural min-max optimization algorithms}

Unlike the offline setting considered so far, where all the input data are available from the outset, in the online setting, input data are streamed to the algorithm sequentially, one at a time. The algorithm must compute the corresponding output before the next input arrives and transmit it downstream. Once transmitted, the output cannot be altered. Moreover, the algorithm cannot store in memory any sizable fraction of past inputs or outputs but only a few, $\mathcal{O}(n k)$, state variables.  

Whereas developing algorithms for the online setting is more challenging than that for the offline, it is necessary both for data analysis and for modeling biological neural networks. The size of modern datasets may exceed that of available RAM and/or the output must be computed before the dataset is fully streamed. Biological neural networks operating on the data streamed by the sensory organs are incapable of storing any significant fraction of it and compute the output on the fly.

\begin{figure}
	\centering
	\includegraphics{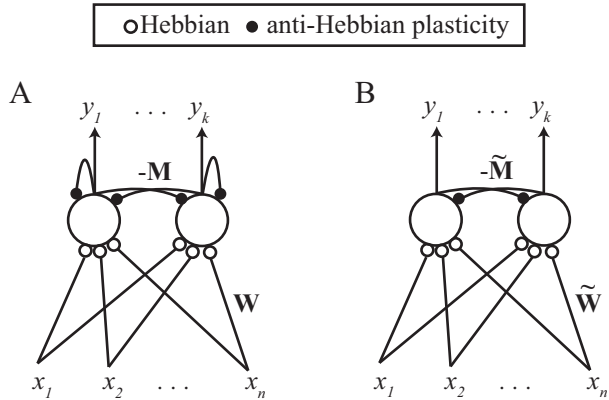}
	\caption{Dimensionality reduction neural networks derived by min-max optimization in the online setting. A. Network with autapses. B. Network without autapses.}
	\label{Fig1}
\end{figure}

\cite{pehlevan2015MDS} gave a derivation of a neural online algorithm for PSP, starting from the original similarity matching cost function \eqref{SM}. Here, instead, we start from the min-max form of similarity matching  \eqref{SMMW2} and end up with a class of algorithms that reduce to the algorithm of \cite{pehlevan2015MDS} for special choices of learning rates. Our main contribution, however, is that the current derivation is much more intuitive and simpler, with insights to why similarity matching leads to local learning rules.

We start by rewriting the min-max PSP objective \eqref{SMMW2} as a sum of time-separable terms that can be optimized independently:
\begin{align}\label{SMMW2t}
\min_{{\bf W}\in \mathbb{R}^{k\times n}}\max_{{\bf M}\in \mathbb{R}^{k\times k}}  \, \frac{1}{T} \sum_{t=1}^T l_{PSP,t} ({\bf W},{\bf M}),
\end{align}  
where
\begin{align}
 l_{PSP,t} ({\bf W},{\bf M}) \equiv 2 {\rm Tr}\left({\bf W}^\top{\bf W}\right) - {\rm Tr}\left({\bf M}^\top{\bf M}\right) +  \min_{{\bf y}_t\in \mathbb{R}^{k\times 1}}  l_t({\bf W},{\bf M},{\bf y}_t),
\end{align}
and
\begin{align}\label{lyapunov}
l_t({\bf W},{\bf M},{\bf y}_t)=-4{\bf x}_t^\top{\bf W}^\top{\bf y}_t + 2{\bf y}_t^\top{\bf M}{\bf y}_t.
\end{align} 
This separation in time is a benefit of the min-max PSP objective \eqref{SMMW2}, and leads to a natural way to derive an online algorithm that was not available for the original similarity matching cost function \eqref{SM}.

To solve the optimization problem, Eq. \eqref{SMMW2t}, in the online setting, we optimize sequentially each ${l}_{PSP,t}$. For each $t$, first, minimize Eq.\eqref{lyapunov} with respect to ${\bf y}_t$ while keeping ${\bf W}_t$ and ${\bf M}_t$ fixed. Second, 
make a gradient descent-ascent step with respect to ${\bf W}_t$ and ${\bf M}_t$ for fixed ${\bf Y}$:
\begin{align}\label{ogda}
{\bf W}_{t+1}&= {\bf W}_t - \eta_t \frac {\partial l_{PSP,t}({\bf W}_t,{\bf M}_t)}{\partial {\bf W}_t}, \nonumber \\
 {\bf M}_{t+1} & = {\bf M}_t  +  \frac{\eta_{t}}{\tau} \frac {\partial l_{PSP,t}({\bf W}_t,{\bf M}_t)}{\partial {\bf M}_t},
\end{align}
where $0<\eta_{t}<1$ is the ${\bf W}$ learning rate and $\tau>0$ is the ratio of ${\bf W}$ and ${\bf M}$ learning rates. As before, Proposition \ref{lem1} (Appendix \ref{chainrule}) ensures that the online gradient descent-ascent updates, Eq. \eqref{ogda}, follow from alternating optimization \citep{olshausen1996emergence,olshausen1997sparse,arora2015simple} of $l_{PSP,t}$.

\begin{algorithm}[H]
\caption{Online min-max PSP\label{onlinePSPalg}}
\begin{algorithmic}[1]
\STATE At $t=0$, initialize the synaptic weight matrices, ${\bf W}_1$ and ${\bf M}_1$. ${\bf M}_1$ must be symmetric and positive definite.
\STATE Repeat for each $t = 1,\ldots T$
\begin{ALC@g}
\STATE Receive input ${\bf x}_t$
\STATE Neural activity:  Run until convergence
\begin{align}\label{gdMDSon}
\frac{d{\bf y}_t(\gamma)}{d\gamma} = {\bf W}_t{\bf x}_t -{\bf M}_t{\bf y}_t.
\end{align}
\STATE Plasticity:
Update synaptic weight matrices,
\begin{align}\label{wmSM}
{\bf W}_{t+1} &= {\bf W}_{t} + 2\eta_{t} \left({\bf y}_t{\bf x}_t^\top-{\bf W}_t\right), \nonumber \\
{\bf M}_{t+1} &= {\bf M}_{t}+\frac{\eta_{t}}{\tau} \left({\bf y}_t{\bf y}_t^\top-{\bf M}_t\right).
\end{align}
\end{ALC@g}
\end{algorithmic}
\end{algorithm}
Algorithm \ref{onlinePSPalg} can be implemented by a biologically plausible neural network. The dynamics \eqref{gdMDSon} corresponds to neural activity in a recurrent circuit, where ${\bf W}_t$ is the feedforward synaptic weight matrix and $-{\bf M}_t$ is the lateral synaptic weight matrix, Fig. \ref{Fig1}A. Since ${\bf M}_t$ is always positive definite, Eq. \eqref{lyapunov} is a Lyapunov function for neural activity. Hence the dynamics is guaranteed to converge to a unique fixed point, ${\bf y}_t = {\bf M}^{-1}_t{\bf W}_t{\bf x}_t$, where matrix inversion is computed iteratively in a distributed manner.

Updates of covariance matrices, Eq. \eqref{wmSM}, can be interpreted as synaptic learning rules: Hebbian for feedforward and anti-Hebbian (due to the $``-"$ sign in \eqref{gdMDSon}) for lateral synaptic weights. Importantly, these rules are local - the weight of each synapse depends only on the activity of the pair of neurons that synapse connects - and therefore biologically plausible. 

Even requiring full optimization with respect to ${\bf y}_t$ vs. a gradient step with respect to ${\bf W}_t$ and ${\bf M}_t$ may have a biological justification. As neural activity dynamics is typically faster than synaptic plasticity, it may settle before the arrival of the next input. 

To see why similarity matching leads to local learning rules let us consider Eqs. \eqref{SMMW2} and \eqref{SMMW2t}. Aside from separating in time, useful for derivation of online learning rules, $L_{PSP}({\bf W},{\bf M},{\bf Y})$ also separates in synaptic weights and their pre- and postsynaptic neural activities,
\begin{align}
L_{PSP}({\bf W},{\bf M},{\bf Y})=\sum_t  \left[\sum_{ij}\left( 2W_{ij}^2-4W_{ij}x_{t,j}y_{t,i}\right) - \sum_{ij}\left( M_{ij}^2+2M_{ij}y_{t,j}y_{t,i}\right)\right].
\end{align}
Therefore, a derivative with respect to a synaptic weight depends only on the quantities accessible to the synapse.

Finally, we address two potential criticisms of the neural PSP algorithm. First is the existence of autapses, i.e. self-coupling of neurons, in our network manifested in nonzero diagonals of the lateral connectivity matrix, ${\bf M}$, Fig \ref{Fig1}A. Whereas autapses are encountered in the brain, they are rarely seen in principal neurons \citep{ikeda2006autapses}. Second is the symmetry of lateral synaptic weights in our network which is not observed experimentally. We derive an autapse-free network architecture (zeros on the diagonal of the lateral synaptic weight matrix ${\bf M}_t$) with asymmetric lateral connectivity, Fig \ref{Fig1}B, by using coordinate descent \citep{pehlevan2015MDS} in place of gradient descent in the neural dynamics stage \eqref{gdMDSon} (see Appendix \ref{coord}). The resulting algorithm produces the same outputs as the current algorithm and for the special case $\tau=1/2$ and $\eta_t = \eta/2$, reduces to the  algorithm with ``forgetting'' of \cite{pehlevan2015MDS}.

\section{From constrained similarity matching to Hebbian/anti-Hebbian networks for PSW}\label{S3}

The variable substitution method we introduced in the previous section can be applied to other computational objectives in order to derive neural networks with local learning rules. To give an example, we derive a neural network for PSW, which can be formulated as a constrained similarity matching problem. This example also illustrates how an optimization constraint can be implemented by biological mechanisms.

 \subsection{Derivation of PSW from constrained similarity matching}
 
 The PSW problem is closely related to PSP: project centered input data samples onto the principal subspace ($k\le n$), and ``spherize" the data in the subspace so that the variances in all directions are 1. To derive a neural PSW algorithm, we use the similarity matching objective with an additional constraint:
 \begin{align}\label{CSM}
 {\rm PSW:} \hspace{3cm} \min_{{\bf Y}\in \mathbb{R}^{k\times T}} \frac 1{T^2} \left\Vert {\bf X}^\top{\bf X}-{\bf Y}^\top{\bf Y}\right\Vert_F^2, \qquad {\rm s.t.}\quad \frac 1T{\bf Y}{\bf Y}^\top = {\bf I}
 \end{align}
 We rewrite Eq. \eqref{CSM} by expanding the Frobenius norm squared and dropping the ${\rm Tr} \left({\bf Y}^\top{\bf Y}{\bf Y}^\top{\bf Y}\right)$ term, which is constant under the constraint, thus reducing \eqref{CSM} to a constrained similarity alignment problem:
 \begin{align}\label{CSA}
 \min_{{\bf Y}\in \mathbb{R}^{k\times T}} \left ( -\frac 1{T^2} {\bf X}^\top{\bf X}{\bf Y}^\top{\bf Y} \right ), \qquad {\rm s.t.}\quad \frac 1T{\bf Y}{\bf Y}^\top = {\bf I}.
 \end{align}
 To see that objective \eqref{CSA} is optimized by the PSW, first, substitute a Singular Value Decomposition (SVD) for matrices ${\bf X}$ and ${\bf Y}$ and note that the alignment is maximized by matching right singular vectors of  ${\bf Y}$ to ${\bf X}$ and rotating to the diagonal basis (for a rigorous proof see \citep{pehlevan2015normative}).  Since the squared singular values of ${\bf Y}$ equal unity, the objective \eqref{CSA} is reduced to a summation of $k$ squared singular values of ${\bf X}$ and is optimized by choosing the top $k$. Then, ${\bf Y}$ is given by the top $k$ right singular vectors of ${\bf X}$ scaled by $\sqrt{T}$. As before, objective \eqref{CSA} is invariant to the left-multiplication of  ${\bf Y}$ by an orthogonal matrix and, therefore, has infinitely many degenerate solutions. 
 
 Next, we derive a mixed PSW objective from Eq. \eqref{CSA} by introducing two new dynamical variable matrices: the input-output correlation matrix, ${\bf W}=\frac 1T {\bf X}{\bf Y}^\top$, and the Lagrange multiplier matrix, ${\bf M}$, for the whitening constraint: 
 \begin{align}
 \min_{{\bf Y}\in \mathbb{R}^{k\times T}} \min_{{\bf W}\in \mathbb{R}^{k\times n}}\max_{{\bf M}\in \mathbb{R}^{k\times k}} \, &L_{PSW}({\bf W},{\bf M},{\bf Y}), 
 \end{align}
 where
 \begin{align}\label{CSMMW}
 \qquad L_{PSW}({\bf W},{\bf M},{\bf Y}) &\equiv  -\frac{2}{T}{\rm Tr}\left({\bf X}^\top{\bf W}^\top{\bf Y}\right) + {\rm Tr}\left({\bf W}^\top{\bf W}\right) +  {\rm Tr}\left({\bf M}  \left(\frac{1}{T}{\bf Y}{\bf Y}^\top -{\bf I}\right)\right).
 \end{align}  
 To see that Eq. \eqref{CSMMW} is equivalent to Eq. \eqref{CSA}, find optimal ${\bf W}^*=\frac 1T {\bf Y}{\bf X}^\top$ by setting the corresponding derivatives of the objective \eqref{CSMMW} to zero. Then, substitute ${\bf W}^*$  into Eq. \eqref{CSMMW} to obtain the Lagrangian of Eq. \eqref{CSA}. 
 
 Finally, we exchange the order of minimization with respect to ${\bf Y}$ and ${\bf W}$ as well as the order of minimization with respect to ${\bf Y}$ and maximization with respect to ${\bf M}$ in Eq. \eqref{CSMMW} (see Proposition \ref{minmaxPSW} in Appendix \ref{A2} for a proof). Then, we arrive at the following min-max optimization problem with a mixed objective:
 \begin{align}\label{CSMMW2}
 \min_{{\bf W}\in \mathbb{R}^{k\times n}}\max_{{\bf M}\in \mathbb{R}^{k\times k}} \min_{{\bf Y}\in \mathbb{R}^{k\times T}}L_{PSW}({\bf W},{\bf M},{\bf Y}) ,
 \end{align}  
 where $L_{PSW}({\bf W},{\bf M},{\bf Y})$ is defined in Eq. \eqref{CSMMW}.

 \subsection{Offline PSW algorithm}
 
 Next, we give an offline algorithm for the PSW problem, using the alternating optimization procedure as before. We solve Eq. \eqref{CSMMW2} by, first, optimizing with respect to ${\bf Y}$ for fixed ${\bf W}$ and ${\bf M}$ and, second, making a gradient descent-ascent step with respect to ${\bf W}$ and ${\bf M}$ while keeping ${\bf Y}$ fixed\footnote{This alternating optimization is identical to a gradient descent-ascent (see Proposition \ref{lem1} in Appendix \ref{chainrule}) in ${\bf W}$ and ${\bf M}$ on the objective:  
 	\begin{align*}
 	l_{PSW}({\bf W},{\bf M}) \equiv   \min_{{\bf Y}\in \mathbb{R}^{k\times T}} L_{PSW}({\bf W},{\bf M},{\bf Y}).
 	\end{align*}
 }.
 We arrive at the following algorithm:  
 \begin{algorithm}[H]
 	\caption{Offline min-max PSW\label{offlinePSWalg}}
 	\begin{algorithmic}[1]
 		\STATE Initialize ${\bf W}$. Initialize ${\bf M}$ as a positive definite matrix.
 		\STATE Iterate until convergence:
 		\begin{ALC@g}
 			\STATE
 			Minimize Eq. \eqref{CSMMW} with respect to ${\bf Y} $, keeping ${\bf W}$ and ${\bf M}$ fixed: 
 			\begin{align}\label{YY}
 			{\bf Y} = {\bf M}^{-1}{\bf W}{\bf X}.
 			\end{align}

 			\STATE  Perform a gradient descent-ascent step with respect to ${\bf W}$ and ${\bf M}$ for a fixed ${\bf Y}$:
 			\begin{align}\label{PSWWMupdate}
 			{\bf W}&\longleftarrow {\bf W} + 2\eta \left(\frac 1T {\bf Y}{\bf X}^\top- {\bf W}\right), \nonumber \\
 			{\bf M} & \longleftarrow {\bf M}  +  \frac{\eta}{\tau} \left(\frac 1T {\bf Y}{\bf Y}^\top-{\bf I}\right).
 			\end{align}
 			where the step size, $0<\eta<1$, may depend on the iteration.
 		\end{ALC@g}
 		
 	\end{algorithmic}
 \end{algorithm}
 Convergence of Algorithm \ref{offlinePSWalg} requires the input covariance matrix, ${\bf C}$, to have at least $k$ non-zero eigenvalues. Otherwise, a consistent solution cannot be found because update \eqref{PSWWMupdate} forces ${\bf Y}$ to be full-rank while Eq. \eqref{YY} lowers its rank.

 \subsection{Linearly stable fixed points of Algorithm \ref{offlinePSWalg} correspond to PSW}
 
 Here we claim that convergence of Algorithm \ref{offlinePSWalg} to fixed ${\bf W}$ and ${\bf M}$ implies PSW of ${\bf X}$. In the limit of small learning rate, the gradient descent-ascent dynamics can be approximated with the system of  differential equations:
 \begin{align}\label{gdPSW}
 {\bf Y}(t) &= {\bf M}^{-1}(t){\bf W}(t){\bf X}, \nonumber \\
 \frac{d {\bf W}(t)}{dt} &= \frac 2{T}{\bf Y}(t){\bf X}^\top -2{\bf W}(t), \nonumber \\
 {\tau}\frac{d{\bf M}(t)}{dt} &= \frac{1}{T}{\bf Y}(t){\bf Y}(t)^\top - {\bf I}(t),
 \end{align}
 where $t$ is now the time index for gradient descent-ascent dynamics. We again define the neural filter matrix ${\bf F} = {\bf M}^{-1}{\bf W}$. 
 
 \begin{Th}\label{mainLSPSW} Fixed points of the dynamical system \eqref{gdPSW} have the following properties: 
 	\begin{enumerate}
 		\item The outputs are whitened, i.e. $\frac 1T {\bf Y}{\bf Y}^\top = {\bf I}$.
 		\item The neural filters span a $k$-dimensional subspace in $\mathbb{R}^n$ which is spanned by some $k$ eigenvectors of the input covariance matrix.
 		\item Stability of the fixed point requires that the neural filters span the {\bf principal} subspace of ${\bf X}$.
 		\item Suppose the neural filters span the principal subspace.  This fixed point is linearly stable if and only if 
 		\begin{align}\label{tauPSW}
 		\tau < \frac{\sigma_i+\sigma_j}{2\left(\sigma_i-\sigma_j\right)^2}
 		\end{align}
 		for all $(i,j)$ pairs, $i\neq j$. By linear stability we mean that linear perturbations of ${\bf W}$ and ${\bf M}$ converge to a rotation of the original neural filters within the principal subspace. 
 	\end{enumerate}
 \end{Th}
 \begin{proof} See Appendix \ref{proofmainLSPSW}. \end{proof}
 
 Based on Theorem \ref{mainLSPSW} we claim that, provided Algorithm \ref{offlinePSWalg} converges, this fixed point corresponds to a PSW of input data. Unlike the PSP case, the neural filters are not orthonormal.

 \subsection{Online algorithm for PSW}
 
 As before, we start by rewriting the min-max PSW objective \eqref{CSMMW2} as a sum of time-separable terms that can be optimized independently:
 \begin{align}\label{CSMMW2t}
 \min_{{\bf W}\in \mathbb{R}^{k\times n}}\max_{{\bf M}\in \mathbb{R}^{k\times k}}  \, \frac{1}{T} \sum_{t=1}^T l_{PSW,t} ({\bf W},{\bf M}).
 \end{align}  
 where
 \begin{align}
 l_{PSW,t} ({\bf W},{\bf M}) \equiv  {\rm Tr}\left({\bf W}^\top{\bf W}\right) - {\rm Tr}\left({\bf M}\right) +  \frac 12 \min_{{\bf y}_t\in \mathbb{R}^{k\times 1}} l_t({\bf W},{\bf M},{\bf y}_t).
 \end{align}
 and $l_t({\bf W},{\bf M},{\bf y}_t)$ is defined in Eq. \eqref{lyapunov}.
 In the online setting, Eq. \eqref{CSMMW2t} can be optimized by sequentially minimizing each ${ l}_{PSW,t}$. For each $t$, first, minimize \eqref{lyapunov} with respect to ${\bf y}_t$ for fixed ${\bf W}_t$ and ${\bf M}_t$, second, update ${\bf W}_t$ and ${\bf M}_t$ according to a gradient descent-ascent step for fixed ${\bf y}_t$:  
 \begin{align}\label{ogda2}
 {\bf W}_{t+1}&= {\bf W}_t - \eta_t \frac {\partial l_{PSW,t}({\bf W}_t,{\bf M}_t)}{{\bf W}_t}, \nonumber \\
 {\bf M}_{t+1} & = {\bf M}_t  +  \frac{\eta_{t}}{\tau} \frac {\partial l_{PSW,t}({\bf W}_t,{\bf M}_t)}{{\bf M}_t},
 \end{align}
 where $0<\eta_{t}<1$ is the ${\bf W}$ learning rate and $\tau>0$ is the ratio of ${\bf W}$ and ${\bf M}$ learning rates.
 
 As before, Proposition \ref{lem1} ensures that the online gradient descent-ascent updates, Eq. \eqref{ogda2}, follow from alternating optimization \citep{olshausen1996emergence,olshausen1997sparse,arora2015simple} of $l_{PSW,t}$. 
 \begin{algorithm}[H]
 	\caption{Online min-max PSW\label{onlinePSWalg}}
 	\begin{algorithmic}[1]
 		\STATE At $t=0$, initialize the synaptic weight matrices, ${\bf W}_1$ and ${\bf M}_1$. ${\bf M}_1$ must be symmetric and positive definite.
 		\STATE Repeat for each $t = 1,\ldots, T$
 		\begin{ALC@g}
 			\STATE Receive input ${\bf x}_t$
 			\STATE Neural activity:  Run until convergence
 			\begin{align}\label{gdCSMon}
 			\frac{d{\bf y}_t(\gamma)}{d\gamma} = {\bf W}_t{\bf x}_t -{\bf M}_t{\bf y}_t.
 			\end{align}
 			\STATE Plasticity:
 			Update synaptic weight matrices,
 			\begin{align}\label{wmCSM}
 			{\bf W}_{t+1} &= {\bf W}_{t} + 2\eta_{W,t} \left({\bf y}_t{\bf x}_t^\top-{\bf W}_t\right), \nonumber \\
 			{\bf M}_{t+1} &= {\bf M}_{t}+\eta_{M,t} \left({\bf y}_t{\bf y}_t^\top-{\bf I}_t\right).
 			\end{align}
 		\end{ALC@g}
 	\end{algorithmic}
 \end{algorithm}
 Algorithm \ref{onlinePSWalg} can be implemented by a biologically plausible single-layer neural network with lateral connections as in Algorithm \ref{onlinePSPalg}, Fig. \ref{Fig1}A. Updates to synaptic weights, Eq. \eqref{wmCSM}, are local, Hebbian/anti-Hebbian plasticity rules. An autapse-free network architecture, Fig \ref{Fig1}B, may be obtained using coordinate descent \citep{pehlevan2015MDS} in place of gradient descent in the neural dynamics stage \eqref{gdCSMon} (see Appendix \ref{coordPSW}). 
 
The lateral connections here are the Lagrange multipliers introduced in the offline problem, Eq. \eqref{CSMMW}. In the PSP network, they resulted from a variable transformation of the output covariance matrix. This difference caries over to the learning rules, where in Algorithm \ref{onlinePSWalg}, the lateral learning rule is enforcing the whitening of the output, but in Algorithm \ref{onlinePSPalg}, the lateral learning rule sets the lateral weight matrix to the output covariance matrix.

\section{Game theoretical interpretation of Hebbian/anti-Hebbian learning}\label{S4}
 
 In the original similarity matching objective, Eq. \eqref{SM}, the only variables are neuronal activities which, at the optimum, represent principal components. In Section \ref{S2}, we rewrote these objectives by introducing matrices {\bf W} and {\bf M} corresponding to synaptic connection weights, Eq. \eqref{SMMW}. Here, we eliminate neural activity variables altogether and arrive at a min-max formulation in terms of feedforward, ${\bf W}$, and lateral, ${\bf M}$, connection weight matrices only. This formulation lends itself to a game-theoretical interpretation. 
 
 Since in the offline PSP setting, optimal ${\bf M}^*$ in Eq. \eqref{SMMW2} is an invertible matrix (because ${\bf M}^*=\frac 1T {\bf Y}^*{{\bf Y}^*}^\top$, see also Appendix \ref{A1}), we can restrict our optimization to invertible matrices, ${\bf M}$, only. Then, we can optimize objective \eqref{SMMW} with respect to ${\bf Y}$ and substitute its optimal value ${\bf Y}^*={\bf M}^{-1}{\bf W}{\bf X}$ into \eqref{SMMW} and \eqref{SMMW2} to obtain:
 \begin{align}\label{SMMW3}
 &\min_{{\bf W}\in \mathbb{R}^{k\times n}}\max_{{\bf M}\in \mathbb{R}^{k\times k}}-\frac{2}{T}{\rm Tr}\left( {\bf X}^\top{\bf W}^\top{\bf M}^{-1} {\bf W}{\bf X}\right) + 2 {\rm Tr}\left({\bf W}^\top{\bf W}\right) -  {\rm Tr}\left({\bf M}^\top{\bf M}\right),\nonumber \\
 &\text{ s.t. ${\bf M}$ is invertible.}
 \end{align}  
 This min-max objective admits a game-theoretical interpretation where feedforward, ${\bf W}$, and lateral, ${\bf M}$, synaptic weight matrices oppose each other. To reduce the objective, feedforward synaptic weight vectors of each output neuron attempt to align with the direction of maximum variance of input data. However, if this was the only driving force then all output neurons would learn the same synaptic weight vectors and represent the same top principal component. At the same time, linear dependency between different feedforward synaptic weight vectors can be exploited by the lateral synaptic weights to increase the objective by cancelling the contributions of different components. To avoid this, the feedforward synaptic weight vectors become linearly independent and span the principal subspace.

 A similar interpretation can be given for PSW, where feedforward, ${\bf W}$, and lateral, ${\bf M}$, synaptic weight matrices oppose each other adversarially. 
 %
 %

\section{Novel formulations of dimensionality reduction using fractional exponents}\label{S5}

In this section, we point to a new class of dimensionality reduction objective functions that naturally follow from the min-max objectives \eqref{SMMW} and \eqref{SMMW2}. Eliminating both the neural activity variables, {\bf Y}, and the lateral connection weight matrix, {\bf M}, we arrive at optimization problems in terms of the feedforward weight matrix, {\bf W}, only. The rows of optimal {\bf W} form a non-orthogonal basis of the principal subspace. Such formulations of principal subspace problems involve fractional exponents of matrices and, to the best of our knowledge, have not been proposed previously.

 By replacing $\max_{\bf M}\min_{\bf Y}$ optimization in the min-max PSP objective, Eq. \eqref{SMMW2}, by its saddle point value (see Proposition \ref{minmaxPSP} in Appendix \ref{A1}) we find the following objective expressed solely in terms of ${\bf W}$:
\begin{align}\label{fractionalPSP}
\min_{{\bf W}\in \mathbb{R}^{k\times n}} {\rm Tr} \left(-\frac{3}{T^{2/3}}\left({\bf W}{\bf X}{\bf X}^\top{\bf W}^\top\right)^{2/3}+2{\bf W}{\bf W}^\top\right),
\end{align}
The rows of the optimal ${\bf W}$ are not principal eigenvectors, rather the rowspace of ${\bf W}$ spans the principal subspace.

By replacing $\max_{\bf M}\min_{\bf Y}$ optimization in the min-max PSW objective, Eq. \eqref{CSMMW2}, by its  optimal value (see Proposition \ref{minmaxPSW} in Appendix \ref{A2}):
\begin{align}\label{fractionalPSW}
\min_{{\bf W}\in \mathbb{R}^{k\times n}}  {\rm Tr} \left(-\frac{2}{T^{1/2}}\left({\bf W}{\bf X}{\bf X}^\top{\bf W}^\top\right)^{1/2}+{\bf W}{\bf W}^\top\right).
\end{align}
As before, the rows of the optimal ${\bf W}$ are not principal eigenvectors, rather the rowspace of ${\bf W}$ spans the principal eigenspace. 

We observe that the only material difference between Eqs. \eqref{fractionalPSP} and \eqref{fractionalPSW} is in the value of the fractional exponent. Based on this, we conjecture that any objective function of such form with a fractional exponent from a continuous range is optimized by ${\bf W}$ spanning the principal subspace. Such solutions would differ in the eigenvalues associated with the corresponding components.  

A supporting argument for our conjecture comes from the work of \cite{miao1998fast}, which studied the cost
\begin{align}\label{logPCA}
\min_{{\bf W}\in \mathbb{R}^{k\times n}}  {\rm Tr} \left(-\log\left({\bf W}{\bf X}{\bf X}^\top{\bf W}^\top\right)+{\bf W}{\bf W}^\top\right).
\end{align}
Eq. \ref{logPCA} can be seen as a limiting case of our conjecture, where the fractional exponent goes to zero. Indeed, \cite{miao1998fast} proved that the rows of optimal ${\bf W}$ are an orthonormal basis for the principal eigenspace.

\section{Numerical experiments}\label{S6}

\begin{figure}
	\centering
	\includegraphics[width = \textwidth]{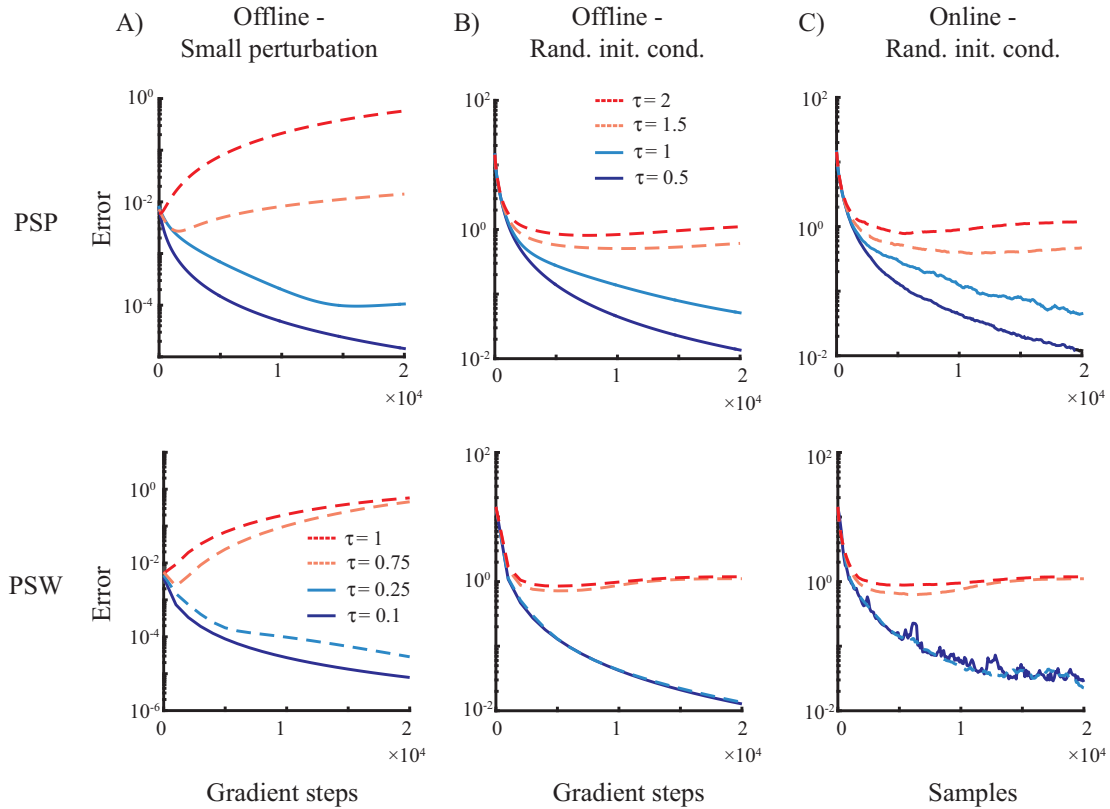}
	\vskip -120pt 
	\caption{Demonstration of the stability of the PSP (top row) and PSW (bottom row) algorithms. We constructed an $n=10$ by $T=2000$ data matrix ${\bf X}$ from its SVD, where the left and right singular vectors are chosen randomly, the top three singular values are set to $\lbrace \sqrt{3T},\sqrt{2T},\sqrt{T}\rbrace$ and the rest of the singular values are chosen uniformly in $[0,0.1\sqrt{T}]$. Learning rates were $\eta_{t} = 1/\left(10^3+t\right)$.  Errors were defined using deviation of the neural filters from their optimal values \citep{pehlevan2015MDS}. Let ${\bf U}$ be the $10\times 3$ matrix whose columns are the top 3 left singular vectors of ${\bf X}$. PSP error: $\left\Vert {\bf F}(t)^\top{\bf F}(t) -{\bf U}{\bf U}^\top\right\Vert_F$, PSW error: $\left\Vert {\bf F}(t)^\top{\bf F}(t) -{\bf U}{\bf S}{\bf U}^\top\right\Vert_F$, with ${\bf S} = {\rm diag}\left([1/3, 1/2, 1]\right)$ in MATLAB notation. Solid (dashed) lines indicate linearly stable (unstable) choices of $\tau$. A) Small perturbations to the fixed point. ${\bf W}$ and ${\bf M}$ matrices were initialized by adding a random Gaussian variable, $\mathcal{N}(0,10^{-6})$, elementwise to their fixed point values.  B) Offline algorithm, initialized with random ${\bf W}$ and ${\bf M}$ matrices. C) Online algorithm, initialized with the same initial condition as in B). A random column of ${\bf X}$ is processed at each time.
		\label{Fig2}}
\end{figure}

Next, we test our findings using a simple artificial dataset. We generated an $n=10$ dimensional dataset and we simulated our offline and online algorithms to reduce this dataset to $k=3$ dimensions, using different values of the parameter $\tau$. The results are plotted in Figs. \ref{Fig2}, \ref{FigM}, \ref{FigTau} and \ref{Figcomp} along with details of the simulations in the figures' caption.

Consistent with Theorems \ref{mainLSPSP} and \ref{mainLSPSW}, small perturbations to PSP and PSW fixed points decayed (solid lines) or grew (dashed lines) depending on the value of  $\tau$, Fig. \ref{Fig2}A. Offline simulations that start from random initial conditions converged to the PSP (or the PSW) solution if the fixed point was linearly stable, Fig. \ref{Fig2}B. Interestingly, the online algorithms' performance were very close to that of the offline, Fig. \ref{Fig2}C.

The error for linearly unstable simulations in Fig. \ref{Fig2} saturates rather than blowing up. This may seem at odds with Theorems \ref{mainLSPSP} and \ref{mainLSPSW}, which stated that if there is a stable fixed point of the dynamics, it should be the PSP/PSW solution. A closer look resolves this dilemma. In  Fig. \ref{FigM}, we plot the evolution of an element of the ${\bf M}$ matrix in the offline algorithms for stable and unstable choices of $\tau$. When the principal subspace is linearly unstable, the synaptic weights exhibit undamped oscillations. The dynamics seems to be confined to a manifold with a fixed distance (in terms of the error metric) from the principal subspace. That the error does not grow to infinity is a result of the stabilizing effect of min-max antagonism of the synaptic weights. Online algorithms behave similarly.

\begin{figure}
	\centering
	\includegraphics[width = \textwidth]{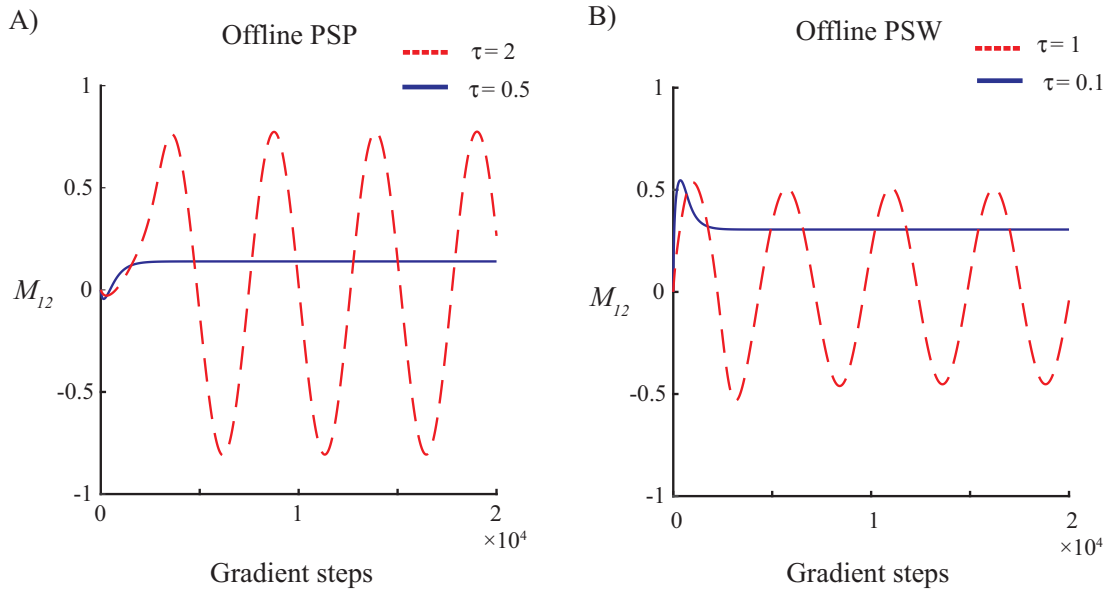}
	\caption{Evolution of a synaptic weight. Same dataset was used as in Fig. \ref{Fig2}. $\eta = 10^{-3}$. \label{FigM}}
\end{figure}

Next, we studied in detail the effect of $\tau$ parameter on the convergence. In the offline algorithm, we plot the error after a fixed number of gradient steps, as a function of $\tau$. For PSP, there is an optimal $\tau$. Decreasing $\tau$ beyond the optimal value doesn't lead to a degradation in performance, however increasing it leads to a rapid increase in the error. For PSW, there is a plateau of low error for low values of $\tau$ but a rapid increase as one approaches the linear instability threshold. Online algorithms behave similarly.

\begin{figure}
	\centering
	\includegraphics[width = \textwidth]{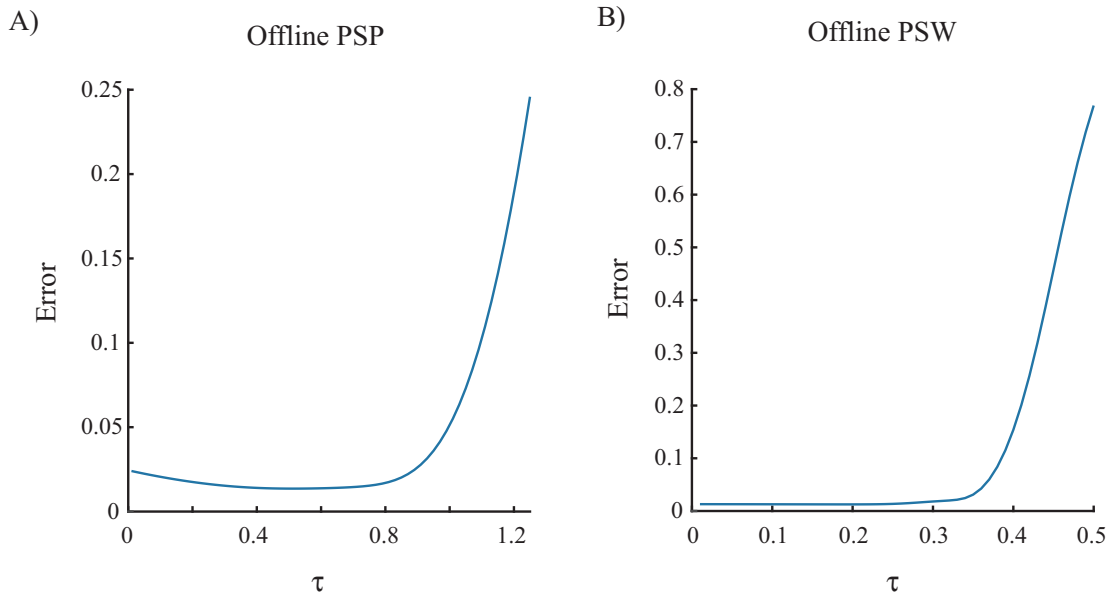}
	\caption{Effect of $\tau$ of performance. Error after $2\times 10^4$ gradient steps are plotted as a function of different choices of $\tau$. Same dataset was used as in Fig. \ref{Fig2} with same network initalization and learning rates. Both curves start from $\tau =0.01$ and go to the maximum value allowed for linear stability. \label{FigTau}}
\end{figure}

Finally, we compared the performance of our online PSP algorithm to neural PSP algorithms with heuristic learning rules such as the Subspace Network \citep{oja1989neural} and the Generalized Hebbian Algorithm (GHA) \citep{sanger1989optimal}, on the same dataset. We found that our algorithm converges much faster (Fig. \ref{Figcomp}). Previously, the original similarity matching network \citep{pehlevan2015MDS}, which is a special case of the online PSP algorithm of this paper, was shown to converge faster than the APEX \citep{kung1994adaptive} and F\"oldiak's  \citep{foldiak1989adaptive} networks.

\begin{figure}
	\centering
	\includegraphics{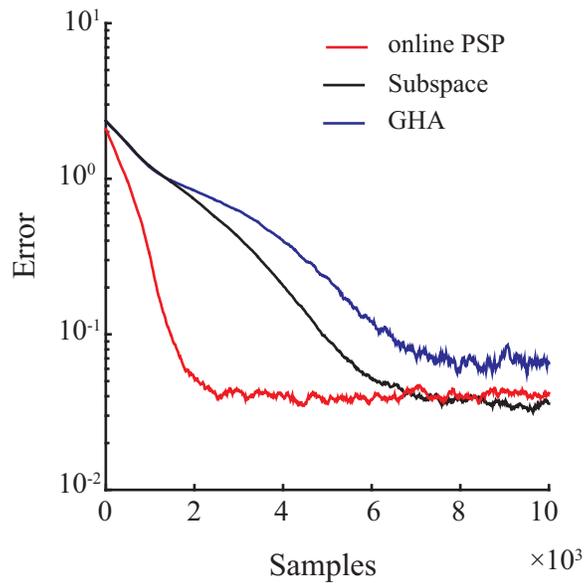}
	\caption{Comparison of the online PSP algorithm with the Subspace Network \citep{oja1989neural} and the GHA \citep{sanger1989optimal}. The dataset and the error metric is as in Fig. \ref{Fig2}. For fairness of comparison, the learning rates in all networks were set to $\eta = 10^{-3}$. $\tau = 1/2$ for the online PSP algorithm \label{Figcomp}. Feedforward connectivity matrices were initialized randomly. For the online PSP algorithm, lateral connectivity matrix was initialized to the identity matrix. Curves show averages over 10 trials.}
\end{figure}

\section{Discussion}

In this paper, through transparent variable substitutions, we demonstrated why biologically plausible neural networks can be derived from similarity matching objectives, mathematically formalized the adversarial relationship between Hebbian feedforward and anti-Hebbian lateral connections using min-max optimization lending itself to a game-theoretical interpretation, and formulated dimensionality reduction tasks as optimizations of fractional powers of matrices. The formalism we developed should generalize to unsupervised tasks other than dimensionality reduction and could provide a theoretical foundation for both natural and artificial neural networks.

In comparing our networks with biological ones, most importantly, our networks rely only on local learning rules that can be implemented by synaptic plasticity. While Hebbian learning is famously observed in neural circuits \citep{bliss1973longa,bliss1973long}, our networks also require anti-Hebbian learning, which can be interpreted as the long-term potentiation of inhibitory postsynaptic potentials. Experimentally, such long-term potentiation can arise from pairing  action potentials in inhibitory neurons with subthreshold
depolarization of postsynaptic pyramidal neurons \citep{komatsu1994age,maffei2006potentiation}. However, plasticity in inhibitory synapses does not have to be Hebbian, i.e. depend on the correlation between pre- and postsynaptic activity  \citep{kullmann2012plasticity}.

To make progress, we had to make several simplifications sacrificing biological realism. In particular, we assumed that neuronal activity is a continuous variable which would correspond to membrane depolarization (in graded potential neurons) or firing rate (in spiking neurons). We ignored the nonlinearity of the neuronal input-output function. Such linear regime could be implemented via a resting state bias (in graded potential neurons) or resting firing rate (in spiking neurons).  

The applicability of our networks as models of biological networks can be judged by experimentally testing the following predictions. First, we predict a relationship between the feedforward and lateral synaptic weight matrices which could be tested using modern connectomics datasets. Second, we suggest that similarity of output activity matches that of the input which could be tested by neuronal population activity measurements using calcium imaging. 

Often the choice of a learning rate is crucial to the learning performance of neural networks. Here, we encountered a nuanced case where the ratio of  feedforward and lateral weights, $\tau$, affects the learning performance significantly. First, there is a maximum value of such ratio, beyond which the principal subspace solution is linearly unstable. The maximum value depends on the principal eigenvalues, but for PSP, $\tau \leq 1/2$ is always linearly stable. For PSW there isn't an always safe choice. Having the same learning rates for feedforward and lateral weights, $\tau=1$, may actually be unstable. Second, linear stability is not the only thing that affects performance. In simulations, for PSP, we observed that there is an optimal value of $\tau$. For PSW, decreasing $\tau$ seems to increase performance until a plateau is reached. This difference between PSP and PSW may be attributed to the difference of origins of lateral connectivity. In PSW algorithms, lateral weights originate from Lagrange multipliers enforcing an optimization constraint. Low $\tau$, meaning higher lateral learning rates, force the network to satisfy the constraint during the whole evolution of the algorithm.

Based on these observation, we can make practical suggestions for the $\tau$ parameter. For PSP, $\tau = 1/2$ seems to be a good choice, which is also preferred from another derivation of an online similarity matching algorithm \citep{pehlevan2015MDS}. For PSW, the smaller the $\tau$, the better it is, although one should make sure that the lateral weight learning rate $\eta/\tau$ is still sufficiently small. 



\subsubsection*{Acknowledgments}
We thank Alex Genkin,  Sebastian Seung, Mariano Tepper and Jonathan Zung for discussions.

\appendix

\pagebreak
\section{Proof of strong min-max property for PSP objective}\label{A1}

Here we show that minimization with respect to ${\bf Y}$ and maximization with respect to ${\bf M}$ can be exchanged in Eq. \eqref{SMMW}. We will make use of the following min-max theorem \citep{boyd2004convex}, for which we give a proof for completeness:
\begin{Th} Let $f:\mathbb{R}^n\times\mathbb{R}^m\longrightarrow \mathbb{R}$. Suppose the saddle-point property holds, i.e. $\exists {\bf a}^* \in \mathbb{R}^n$, ${\bf b}^*\in \mathbb{R}^m$ such that $\forall {\bf a} \in \mathbb{R}^n$, ${\bf b}\in \mathbb{R}^m$
\begin{align}\label{sp}
f({\bf a}^*, {\bf b}) \leq f({\bf a}^*, {\bf b}^*) \leq f({\bf a}, {\bf b}^*).
\end{align}
 Then,
\begin{align}
\max_{{\bf b}} \min_{{\bf a}} f({\bf a},{\bf b}) = \min_{\bf a} \max_{\bf b}  f({\bf a},{\bf b}) = f({\bf a}^*, {\bf b}^*) .  
\end{align}
\end{Th}
\begin{proof} 
$\forall{\bf c} \in \mathbb{R}^n$, $\min_{{\bf a}}\max_{{\bf b}} f({\bf a},{\bf b}) \leq \max_{{\bf b}} f({\bf c}, {\bf b})$, which implies
\begin{align}
\min_{{\bf a}}\max_{{\bf b}} f({\bf a},{\bf b}) \leq \max_{{\bf b}} f({\bf a}^*, {\bf b}) =  f({\bf a}^*, {\bf b}^*)= \min_{{\bf a}} f({\bf a}, {\bf b}^*) \leq  \max_{{\bf b}}  \min_{{\bf a}} f({\bf a},{\bf b}). 
\end{align}
Since $\max_{{\bf b}}  \min_{{\bf a}} f({\bf a}, {\bf b}) \leq  \min_{{\bf a}}\max_{{\bf b}} f({\bf a},{\bf b})$ is always true, we get an equality. 
\end{proof}

Now, we present the main result of this section.
\begin{prop}\label{minmaxPSP}Define
\begin{align}\label{fMW}
f\left({\bf Y},{\bf M},{\bf A}\right):={\rm Tr}\left( -\frac 4T{\bf A}^\top{\bf Y} + \frac 2T{\bf Y}^\top{\bf M}{\bf Y} \right) - \rm{Tr}\left({\bf M}^\top{\bf M}\right),
\end{align}  
where ${\bf Y}$, ${\bf M}$ and ${\bf A}$ are arbitrary sized, real-valued matrices. $f$ obeys a strong min-max property:
\begin{align}
 \min_{\bf Y}\max_{\bf M}f\left({\bf Y},{\bf M},{\bf A}\right) =  \max_{\bf M}\min_{\bf Y} f\left({\bf Y},{\bf M},{\bf A}\right) = -   \frac{3}{T^{2/3}}\rm{Tr}\left(\left({\bf A}{\bf A}^\top\right)^{2/3}\right) .
\end{align}
\end{prop}
\begin{proof}
We will show that the saddle-point property holds for Eq. \eqref{fMW}. Then the result follows from Theorem 1.

If the saddle point exists,  it is when $\nabla f = 0$,
\begin{align}\label{sol}
{\bf M}^* &= \frac 1T {\bf Y}^*{{\bf Y}^*}^\top, \nonumber \\
 {\bf M}^*{\bf Y}^* &= {\bf A}.
\end{align}
Note that ${\bf M}^*$ is symmetric and positive semidefinite. Multiplying the first equation by ${\bf M}^*$ on the left and the right, and using the the second equation, we arrive at
\begin{align}
{{\bf M}^*}^3 = \frac 1T{\bf A}{\bf A}^\top.
\end{align}
Solutions to Eq. \eqref{sol} are not unique, because ${\bf M}^*$ may not be invertible depending on ${\bf A}$. However, all solutions give the same value of $f$:
\begin{align}\label{val}
f\left( {\bf Y}^*, {\bf M}^*,{\bf A}\right) &= {\rm Tr}\left( -\frac 4T{\bf A}^\top {\bf Y}^* + \frac 2T{{\bf Y}^*}^\top {\bf M}^* {\bf Y}^* \right) - \rm{Tr}\left({{\bf M}^*}^2\right) \nonumber \\
&={\rm Tr}\left( -\frac 4T{{\bf Y}^*}^\top{\bf M}^*{\bf Y}^* + \frac 2T{{\bf Y}^*}^\top{\bf M}^*{\bf Y}^* \right) - \rm{Tr}\left({{\bf M}^*}^2\right) \nonumber \\
&=- 3{\rm Tr}\left({{\bf M}^*}^2\right)= - \frac{3}{T^{2/3}}{\rm Tr}\left(\left({\bf A}{\bf A}^\top\right)^{2/3}\right).
\end{align}

Now, we check if the saddle-point property, Eq. \eqref{sp}, holds. The first inequality is satisfied:
\begin{align}\label{ineq1}
f\left( {\bf Y}^*, {\bf M}^*,{\bf A}\right) - f\left( {\bf Y}^*, {\bf M},{\bf A}\right) &= \rm{Tr} \left(\frac 2T {{\bf Y}^*}^\top \left({\bf M}^*-{\bf M}\right){\bf Y}^*\right) - \rm{Tr}\left({{\bf M}^*}^2\right) + \rm{Tr}\left({\bf M}^\top{\bf M}\right)\nonumber \\
&=-2\,\rm{Tr} \left({\bf M}^*{\bf M}\right) + \rm{Tr}\left({{\bf M}^*}^2\right) + \rm{Tr}\left({\bf M}^\top{\bf M}\right)\nonumber \\
&= \left\Vert {{\bf M}^*}-{\bf M}\right\Vert_F^2 \geq 0.
\end{align}
The second inequality is also satisfied:
\begin{align}\label{ineq2}
f\left({\bf Y},{\bf M}^*,{\bf A}\right) - f\left( {\bf Y}^*,{\bf M}^*,{\bf A}\right) &= {\rm Tr}\left(-\frac 4T {\bf A}^\top\left({\bf Y}-{\bf Y}^*\right)+ \frac 2T{\bf Y}^\top {\bf M}^*{\bf Y} -\frac 2T {{\bf Y}^*}^\top {\bf M}^*{\bf Y}^* \right) \nonumber \\
&=  {\rm Tr}\left(-\frac 4T {{\bf Y}^*}^\top{\bf M}^*{\bf Y}+\frac 2T{\bf Y}^\top {\bf M}^*{\bf Y} +\frac 2T{{\bf Y}^*}^\top{\bf M}^*{\bf Y}^* \right) \nonumber \\
&=  \frac 2T\,{\rm Tr}\left(\left({\bf Y}-{\bf Y}^*\right)^\top {\bf M}^* \left({\bf Y}-{\bf Y}^*\right) \right) \geq 0,
 \end{align}
where the last line follows form ${\bf M}^*$ being positive semidefinite.

Eq.s \eqref{ineq1} and \eqref{ineq2} show that the saddle-point property \eqref{sp} holds, and therefore $\max$ and $\min$ can be exchanged and the value of $f$ at the saddle-point is $f\left({\bf Y}^*,{\bf M}^*,{\bf A}\right)=- \frac{3}{T^{2/3}}\rm{Tr}\left(\left({\bf A}{\bf A}^\top\right)^{2/3}\right)$.
\end{proof}

\pagebreak

\section{Taking a derivative using a chain rule}\label{chainrule}

\begin{prop} \label{lem1} Suppose a differentiable, scalar function $H({\bf a}_1,\ldots,{\bf a}_m)$, where ${\bf a}_i \in \mathbb{R}^{d_i}$ with arbitrary $d_i$.  Assume a finite minimum with respect to ${\bf a}_m$ exists for a given set of $\lbrace{\bf a}_1,\ldots,{\bf a}_{m-1}\rbrace$:
	\begin{align}
	h({\bf a}_1,\ldots,{\bf a}_{m-1}) = \min_{{\bf a}_m} H({\bf a}_1,\ldots,{\bf a}_m),
	\end{align}
	and the optimal ${\bf a}_m^* = \argmin_{{\bf a}_m} H({\bf a}_1,\ldots,{\bf a}_m)$ is a stationary point
	\begin{align}\label{st}
	\left. \frac{\partial H}{\partial {\bf a}_m}\right |_{\lbrace{\bf a}_1,\ldots,{\bf a}_{m-1},{\bf a}_m^*\rbrace} = 0.
	\end{align}
	Then, for $i=1,\ldots, m-1$
	\begin{align}
	\left. \frac{\partial h}{\partial {\bf a}_i}\right|_ {\lbrace{\bf a}_1,\ldots,{\bf a}_{m-1}\rbrace} = \left.\frac{\partial H}{\partial {\bf a}_i}\right |_{\lbrace{\bf a}_1,\ldots,{\bf a}_{m-1},{\bf a}_m^*\rbrace} 
	\end{align}
\end{prop}
\begin{proof} The result follows from application of the chain rule and the stationarity of the minimum:
	\begin{align}
	\left. \frac{\partial h}{\partial {\bf a}_i}\right|_ {\lbrace{\bf a}_1,\ldots,{\bf a}_{m-1}\rbrace} = \left.\frac{\partial H}{\partial {\bf a}_i}\right |_{\lbrace{\bf a}_1,\ldots,{\bf a}_{m-1},{\bf a}_m^*\rbrace} +\left(\left.\frac{\partial H}{\partial {\bf a}_m}\right |_{\lbrace{\bf a}_1,\ldots,{\bf a}_{m-1},{\bf a}_m^*\rbrace} \right)^\top \left. \frac{\partial {\bf a}_m}{\partial {\bf a}_i}\right|_ {\lbrace{\bf a}_1,\ldots,{\bf a}_{m-1}\rbrace}
	\end{align}
	where the second term is zero due to Eq. \eqref{st}.
\end{proof}

\pagebreak

\section{Proof of Theorem \ref{mainLSPSP}}\label{LSPSP}

Here we prove Theorem \ref{mainLSPSP} using methodology from \citep{pehlevan2015MDS}. 

The fixed points of Eq. \eqref{gdPSP} are ( using $\bar{}$ for fixed point):
\begin{align}\label{fp}
\bar {\bf W} =  \bar {\bf F}{\bf C}, \qquad \qquad\bar {\bf M}  = \bar {\bf F}{\bf C}\bar{\bf F}^\top,
\end{align}
where ${\bf C}$ is the input covariance matrix defined as in Eq. \eqref{Cdef}.

\subsection{Proof of item 1}

The result follows from Eq.s \eqref{fdef} and \eqref{fp}:
\begin{align}\label{orth}
{\bf I} = \bar {\bf M}^{-1} \bar {\bf M}=  \bar {\bf M}^{-1}   \bar {\bf F}{\bf C}\bar{\bf F}^\top  =   \bar {\bf M}^{-1}   \bar {\bf W} \bar{\bf F}^\top = \bar{\bf F} \bar{\bf F}^\top 
\end{align}

\subsection{Proof of item 2}\label{1i2}

First note that at fixed points, $\bar{\bf F}^\top\bar{\bf F}$ and ${\bf C}$ commute:
\begin{align}
\bar{\bf F}^\top\bar{\bf F} {\bf C}= {\bf C} \bar{\bf F}^\top\bar{\bf F}.
\end{align}

\begin{proof}
	The result follows from Eq.s \eqref{fdef} and \eqref{fp}:
	\begin{align}
	\bar{\bf F}^\top\bar{\bf F} {\bf C} = \bar{\bf F}^\top\bar{\bf W} =  \bar{\bf F}^\top\bar{\bf M}\bar{\bf F} = \bar {\bf W}^\top\bar{\bf F} = {\bf C}\bar{\bf F}^\top\bar{\bf F}. 
	\end{align}
\end{proof}

$\bar{\bf F}^\top\bar{\bf F}$ and ${\bf C}$ share the same eigenvectors, because they commute. Orthonormality of neural filters, Eq. \eqref{orth}, implies that the $k$ rows of $\bar {\bf F}$ are degenerate eigenvectors of $\bar{\bf F}^\top\bar{\bf F}$ with unit eigenvalue. To see this: $\left(\bar{\bf F}^\top\bar{\bf F}\right)\bar{\bf F}^\top = \bar{\bf F}^\top$. Because the filters are degenerate, the corresponding $k$ shared eigenvectors of ${\bf C}$ may not be the filters themselves but linear combinations of them. Nevertheless, the shared eigenvectors composed of filters span the same space as the filters. 

Since we are interested in PSP, it is desirable that it is the top $k$ eigenvectors of ${\bf C}$ that spans the filter space. A linear stability analysis around the fixed point reveals that any other combination is unstable, and that the PS is stable if $\tau$ is chosen appropriately.

\subsection{Proof of item 3}\label{1i3}

\subsubsection*{Preliminaries}
In order to perform a linear stability analysis, we linearize the system of equations \eqref{gdPSP} around the fixed point. Even though Eq. \eqref{gdPSP} depends on ${\bf W}$ and ${\bf M}$, we will find it convenient to change variables and work with ${\bf F}$ and ${\bf M}$ instead. 

Using the relation ${\bf F} = {\bf M}^{-1}{\bf W}$, one can express linear perturbations of ${\bf F}$ around its fixed point,  ${\bf \delta F}$, in terms of perturbations of ${\bf W}$ and ${\bf M}$:
\begin{align}
{\bf \delta F} &=\delta\left( {\bf  M}^{-1}\right)\bar{\bf W} + \bar {\bf M}^{-1}{\bf \delta W} = -\bar{\bf M} ^{-1}{\delta \bf M} \bar{\bf F} + \bar {\bf M}^{-1}{\bf \delta W}
\end{align}
Linearization of Eq. \eqref{gdPSP} gives:
\begin{align}
\frac{d\delta{\bf W}}{dt} &= 2 {\bf \delta F}{\bf C}  -2{\bf \delta W},
\end{align}
and
\begin{align}\label{evolM}
\tau\frac{d\delta{\bf M}}{dt} &=  {\bf \delta F}{\bf C} \bar{\bf F}^\top +   \bar {\bf F}{\bf C} \bar{\bf \delta F}^\top -{\bf \delta M}.
\end{align}
Using these, we arrive at:
\begin{align}\label{evol}
\frac{d\delta{\bf F}}{dt} = -\frac 1{\tau} \bar {\bf M}^{-1}\left( {\bf \delta F}{\bf C} \bar{\bf F}^\top +   \bar {\bf F}{\bf C} \bar{\bf \delta F}^\top+(2\tau-1){\bf \delta M}\right)\bar{\bf F} + 2\bar {\bf M}^{-1} {\bf \delta F}{\bf C} -2{\bf\delta  F}.
\end{align}
Eq.s \eqref{evolM} and \eqref{evol} define a closed system of equations.

It will be useful to decompose ${\bf \delta F}$ into components\footnote{see Lemma 3 in \citep{pehlevan2015MDS} for a proof of why such a decomposition always exists.}:
\begin{align}
{\bf \delta F} = {\bf \delta A}\bar {\bf F} + {\bf \delta S}\bar {\bf F} + {\bf \delta B}\bar{\bf G}
\end{align}
where $\delta{\bf A}$ is an $k\times k$ anti-symmetric matrix, $\delta{\bf S}$ is an $k\times k$ symmetric matrix and $\delta{\bf B}$ is an $k\times(n-k)$ matrix. $\bar{\bf G}$ is an $(n-k) \times n$ matrix with orthonormal rows, which are orthogonal to the rows of $\bar{\bf F}$.  ${\bf \delta A}$ and ${\bf \delta S}$ are perturbations that keep the neural filters within the filter space. Anti-symmetric ${\bf \delta A}$ corresponds to rotations of filters within the filter space, preserving orthonormality. Symmetric ${\bf \delta S}$ destroys orthonormality. $\delta {\bf B}$ is a perturbation that takes the neural filters outside of the filter space. 

Let ${\bf v}_{1,\ldots,n}$ be the eigenvectors ${\bf C}$ and $\sigma_{1,\ldots,n}$ be the corresponding eigenvalues. We label them such that  $\bar{\bf F}$ spans the same space as the space spanned by the first $k$ eigenvectors. We choose rows of $\bar{\bf G}$ to be the remaining eigenvectors, i.e. $\bar{\bf G}^\top:=[{\bf v}_{k+1},\ldots,{\bf v}_{n}]$. Note that, with this choice,
\begin{align}\label{G}
\sum_k C_{ik}\bar G^\top_{kj} =  \sigma_{j+m} \bar G^\top_{ij}.
\end{align}

\subsubsection*{Proof}

The proof of item 3 in Theorem \ref{mainLSPSP} follows from studying the stability of ${\bf \delta B}$ component.

Multiplying Eq. \eqref{evol} on the right by $\bar{\bf G}^\top$, one arrives at a decoupled equation for ${\bf \delta B}$:
\begin{align}
&\frac {d \delta{B}_{i}^j}{dt} = \sum_m P^j_{im}\delta B_{m}^j, \qquad \quad P^j_{im}:= 2\left(\bar M^{-1}_{im}\sigma_{j+k} -\delta_{im}\right), 
\end{align}
where for convenience we changed our notation to $\delta B_{kj}=\delta B_{k}^j$. For each $j$, the dynamics is linearly stable if all eigenvalues of all ${\bf P^j}$ are negative. In turn, this implies that for stability, eigenvalues of $\bar{\bf M}$ should be greater than  $\sigma_{k+1,\ldots,n}$.

Eigenvalues of $\bar {\bf M}$ are:
\begin{align}
\text{eig}(\bar{\bf M}) = \left\lbrace \sigma_{1}, \ldots,\sigma_{k}\right\rbrace.
\end{align}
\begin{proof} The eigenvalue equation
	\begin{align}\label{eigEq}
	\bar{\bf F}{\bf C}\bar{\bf F}^\top{\boldsymbol \lambda} &= \lambda {\boldsymbol \lambda} 
	\end{align}
	implies that
	\begin{align}\label{eq1}
	{\bf C}\left(\bar{\bf F}^\top {\boldsymbol \lambda}\right)=\lambda \left(\bar{\bf F}^\top{\boldsymbol \lambda}\right),
	\end{align}
	which can be seen by multiplying Eq. \eqref{eigEq} on the left by $\bar{\bf F}^\top$, using the commutation of  $\bar{\bf F}^\top\bar{\bf F}$ and ${\bf C}$, and the orthonormality of neural filters. Further,  orthonormality of neural filters implies:
	\begin{align}\label{eq2}
	\bar{\bf F}^\top\bar{\bf F} \left(\bar{\bf F}^\top{\boldsymbol \lambda}\right)= \left(\bar{\bf F}^\top{\boldsymbol \lambda}\right).
	\end{align}
	Then, $\left(\bar{\bf F}^\top{\boldsymbol \lambda}\right)$ is a shared eigenvector\footnote{One might worry that $\left(\bar{\bf F}^\top{\boldsymbol \lambda}\right)={\bf 0}$, but this would require $\bar{\bf F}\left(\bar{\bf F}^\top{\boldsymbol \lambda}\right)={\boldsymbol \lambda}={\bf 0}$, which is a contradiction.} between ${\bf C}$ and $\bar{\bf F}^\top\bar{\bf F} $. Shared eigenvectors of ${\bf C}$ with unit eigenvalue in $\bar{\bf F}^\top\bar{\bf F}$ are ${\bf v}_1,\ldots,{\bf v}_k$.  Since the eigenvalue of $\left(\bar{\bf F}^\top{\boldsymbol \lambda}\right)$  with respect to $\bar{\bf F}^\top\bar{\bf F} $ is  1, $\bar{\bf F}^\top{\boldsymbol \lambda} $ must be one of ${\bf v}_1,\ldots,{\bf v}_k$. Then Eq. \eqref{eq1} implies that $\lambda = \left\lbrace \sigma_{1}, \ldots,\sigma_{k}\right\rbrace$ and 
	\begin{align}
	\text{eig}(\bar{\bf M}) = \left\lbrace \sigma_{1}, \ldots,\sigma_{k}\right\rbrace.
	\end{align}

\end{proof}

Then, it follows that linear stability requires
\begin{align}
\left\lbrace \sigma_{1}, \ldots,\sigma_{k}\right\rbrace > \left\lbrace \sigma_{k+1}, \ldots,\sigma_{n}\right\rbrace.
\end{align}
This proves our claim that if at the fixed point, the neural filters span a subspace other than the principal subspace, the fixed point is linearly unstable.

\subsection{Proof of item 4}

We now assume that the fixed point is the principal subspace. From item 3, we know that the ${\bf \delta B}$ perturbations are stable. The proof of item 4 in Theorem \ref{mainLSPSP}, follows from the linear stabilities of  $\delta {\bf A}$ and $\delta {\bf S}$. 

Multiplying Eq. \eqref{evol} on the right by $\bar{\bf F}^\top$, 
\begin{align}\label{l1}
\frac{d {\bf \delta A}} {dt} + \frac{d {\bf \delta S}} {dt} = &\left(2-\frac 1{\tau}\right) \left( \bar{\bf M}^{-1}\left({\bf \delta A}+{\bf \delta S}\right)\bar{\bf M}-\bar{\bf M}^{-1}{\bf \delta M}-{\bf \delta A}\right)-\left(2+\frac{1}{\tau}\right){\bf \delta S}.
\end{align}
Unlike the case of ${\bf \delta B}$, this equation is coupled to  ${\bf \delta M}$, whose dynamics, Eq. \eqref{evolM}, reduces to:
\begin{align}\label{l2}
{\tau}\frac{d {\bf \delta M}} {dt}  = \left({\bf \delta A}+{\bf \delta S}\right)\bar{\bf M}+\bar{\bf M}\left(-{\bf \delta A}+{\bf \delta S}\right)-{\bf \delta M}.
\end{align}
We will only consider symmetric ${\bf \delta M}$ perturbations, although if antisymmetric perturbations were allowed, they would stably decay to zero, because the only antisymmetric term on the RHS of \eqref{l2} would come from ${\bf \delta M}$.

From Eq.s \eqref{l1} and \eqref{l2}, it follows that
\begin{align}\label{l3}
\frac{d } {dt} \left({\bf \delta A} + {\bf \delta S} - \left(2\tau-1\right) \bar{\bf M}^{-1}{\bf \delta M} \right) = -4{\bf \delta S}. 
\end{align}
The RHS is symmetric. Therefore, the antisymmetric part of the LHS must equal zero. This gives us an integral of the dynamics
\begin{align}\label{l4}
{\bf \Omega}:={\bf \delta A}(t)  -  \left(\tau-\frac{1}{2}\right) \left(\bar{\bf M}^{-1}{\bf \delta M}(t)-{\bf \delta M}(t)\bar{\bf M}^{-1}\right),
\end{align}
where ${\bf \Omega}$ is a constant, skew symmetric matrix. This reveals an interesting point, after the perturbation ${\bf \delta A}$ and ${\bf \delta M}$ will not decay to ${\bf 0}$, even if the fixed point is stable. In hindsight, this is expected because due to the symmetry of the problem: there is a manifold of stable fixed points (bases in principal subspace), and perturbations within this manifold should not decay. A similar situation was observed in \citep{pehlevan2015MDS}.

The symmetric part of Eq. \eqref{l3} gives,
\begin{align}
\frac{d } {dt} \left( {\bf \delta S} - \left(\tau-\frac{1}{2}\right) \left(\bar{\bf M}^{-1}{\bf \delta M} + {\bf \delta M}\bar{\bf M}^{-1} \right)\right) = -4{\bf \delta S}, 
\end{align}
which, using \eqref{l2}, implies
\begin{align}\label{l2new}
\frac{d  {\bf \delta S}} {dt} = &\left(1-\frac{1}{2\tau}\right)\left[\bar{\bf M}^{-1}{\bf \delta A}\bar{\bf M}-\bar{\bf M}{\bf \delta A}\bar{\bf M}^{-1}\right] \nonumber \\
&+ \left(1-\frac{1}{2\tau}\right)\left[\bar{\bf M}^{-1} {\bf \delta S} \bar{\bf M}+\bar{\bf M}{\bf \delta S}\bar{\bf M}^{-1}+2{\bf \delta S} \right]  -4{\bf \delta S}  \nonumber \\
& - \left(1-\frac{1}{2\tau}\right)\left(\bar{\bf M}^{-1} {\bf \delta M}+{\bf \delta M}\bar{\bf M}^{-1}\right). 
\end{align}
To summarize, we analyze the linear stability of the system of equations, defined by Eq.s \eqref{l2}, \eqref{l4}, \eqref{l2new}.

Next, we change to a basis where $\bar{\bf M}$ is diagonal. $\bar{\bf M}$ is symmetric, its eigenvalues are the principal eigenvectors $\lbrace \sigma_1, \ldots, \sigma_k\rbrace$ as shown in Appendix \ref{1i3} and it has an orthonormal set of eigenvectors. Let ${\bf U}$ be the matrix that contains the eigenvectors of $\bar{\bf M}$ in its columns.  Define
\begin{align}
{\bf \delta A}^U &:= {\bf U}^\top{\bf \delta A}{\bf U}, \nonumber \\
{\bf \delta S}^U &:= {\bf U}^\top{\bf \delta S}{\bf U}, \nonumber \\
{\bf \delta M}^U &:= {\bf U}^\top{\bf \delta M}{\bf U},\nonumber \\
{\bf \Omega}^U &:= {\bf U}^\top{\bf \Omega}{\bf U}
\end{align}
Expressing Eq.s \eqref{l2}, \eqref{l4}, \eqref{l2new} in this new basis,  in component form, and eliminating ${ \delta A}^U_{ij}$:
\begin{align}\label{2dcomp}
\frac{d} {dt}  \left[\begin{array}{cc}  {\delta M}^U_{ij} \\  {\delta S}^U_{ij}\end{array}\right] = {\bf H}^{ij} \left[\begin{array}{cc}  {\delta M}^U_{ij} \\  {\delta S}^U_{ij}\end{array}\right] + \left[\begin{array}{cc}  \frac 1{\tau} \left(\sigma_j-\sigma_i\right) \\ \left(1-\frac{1}{2\tau}\right)\left(\frac{\sigma_j}{\sigma_i}- \frac{\sigma_i}{\sigma_j} \right)\end{array}\right]{\Omega}^U_{ij}
\end{align}
where
\begin{align}
{\bf H}^{ij} := \left[\begin{array}{cc} 
\left(1-\frac{1}{2\tau}\right)  \left(\sigma_j-\sigma_i\right)\left( \frac{1}{\sigma_i}-\frac 1{\sigma_j}\right) -\frac 1{\tau} & \frac 1{\tau} \left(\sigma_j+\sigma_i\right) \\
\left(1-\frac{1}{2\tau}\right)\left[\left(\frac{\sigma_j}{\sigma_i}- \frac{\sigma_i}{\sigma_j} \right)\left(\tau-\frac{1}{2}\right) \left( \frac{1}{\sigma_i}-\frac 1{\sigma_j}\right)  -  \left( \frac{1}{\sigma_i}+\frac 1{\sigma_j}\right) \right]  &\left(1-\frac{1}{2\tau}\right)\left(\frac{\sigma_j}{\sigma_i}+ \frac{\sigma_i}{\sigma_j}+2 \right)  -4 \\
\end{array}\right]
\end{align}

This is a closed system of equations for each $(i,j)$ pair! The fixed point of this system of equations is at
\begin{align}
{\delta S}^U_{ij} &= 0, \nonumber \\
{\delta M}^U_{ij} & = \frac{{\Omega}^U_{ij} }{\frac 1{\sigma_j-\sigma_i}-\left(\tau-\frac{1}{2}\right) \left( \frac{1}{\sigma_i}-\frac 1{\sigma_j}\right)}.
\end{align}
Hence, if the linear perturbations are stable, the perturbations that destroy the orthonormality of neural filters will decay to zero, and orthonormality will be restored. 

The stability of the fixed point is governed by the trace and the determinant of the matrix ${\bf H}^{ij}$. The trace is 
\begin{align}
{\rm Tr}({\bf H}^{ij}) = -4 + \left(2-\frac 1{\tau} \right)\left(\frac{\sigma_i}{\sigma_j}+\frac{\sigma_j}{\sigma_i},\right) - \frac 1{\tau}
\end{align}
and the determinant is
\begin{align}
\det({\bf H}^{ij}) = 8 + \left(\frac 2{\tau} -4\right)\left(\frac{\sigma_i}{\sigma_j}+\frac{\sigma_j}{\sigma_i}\right).
\end{align}
The system \eqref{2dcomp} is linearly stable if both the trace is negative and the determinant is positive. 

Defining the following function of covariance eigenvalues:
\begin{align}
\gamma_{ij} := \left(\frac{\sigma_i}{\sigma_j}+\frac{\sigma_j}{\sigma_i}\right)=2+\frac{\left(\sigma_i-\sigma_j\right)^2}{\sigma_i\sigma_j},
\end{align}
the trace is negative if and only if
\begin{align}\label{scond1}
\tau<\frac{1+1/\gamma_{ij}}{2-4/\gamma_{ij}}
\end{align}
The determinant is positive if and only if
\begin{align}\label{scond2}
\tau<\frac{1}{2-4/\gamma_{ij}}
\end{align}
Since $\gamma_{ij}>0$, Eq. \eqref{scond2} implies Eq. \eqref{scond1}. For stability, Eq. \eqref{scond2} has to be satisfied for all $(i,j)$ pairs. When $i=j$, $\gamma_{ii}=2$, Eq. \eqref{scond2} is satisfied because RHS is infinity. When $i\neq j$, Eq. \eqref{scond2} is nontrivial, and depends on relations between covariance eigenvalues. Since $\gamma_{ij}\geq 2$, $\tau \leq 1/2$ is always stable.

Collectively, our results prove item 4 of Theorem \ref{mainLSPSP}.

\pagebreak

\section{Proof of strong min-max property for PSW objective}\label{A2}

Here we show that minimization with respect to ${\bf Y}$ and maximization with respect to ${\bf M}$ can be exchanged in Eq. \eqref{CSMMW}. We do this by explicitly calculating the value of
\begin{align}\label{CSMM}
-\frac 2T{\rm Tr}\left( {\bf X}^\top{\bf W}^\top{\bf Y} \right)+{\rm Tr}\left( {\bf M}\left(\frac 1T{\bf Y}{\bf Y}^\top-{\bf I}\right) \right)
\end{align}
with respect to min-max and max-min optimizations, and showing that the value does not change. 

\begin{prop}\label{minmax} Let ${\bf A}\in \mathbb{R}^{k\times T}$ with $k\leq T$. Then
\begin{align}\label{minmaxg}
\min_{{\bf Y}\in \mathbb{R}^{k\times T}}\max_{{\bf M}\in \mathbb{R}^{k\times k}} -\frac 2T{\rm Tr}\left( {\bf A}^\top{\bf Y} \right)+{\rm Tr}\left( {\bf M}\left(\frac 1T{\bf Y}{\bf Y}^\top-{\bf I}\right) \right) =-\frac 2{T^{1/2}}{\rm Tr}\left(\left({\bf A}{\bf A}^\top\right)^{1/2}\right).
\end{align}
\end{prop}
\begin{proof}
Left side of Eq. \eqref{minmaxg} is a constrained optimization problem:
\begin{align}
\min_{{\bf Y}\in \mathbb{R}^{k\times T}} -\frac 2T{\rm Tr}\left( {\bf A}^\top{\bf Y} \right)\qquad {\rm s.t.}\quad \frac 1T {\bf Y}{\bf Y}^\top={\bf I}.
\end{align}
Suppose an SVD of ${\bf A}=\sum_{i=1}^k \sigma_{A,i}{\bf u}_{A,i}{\bf v}^\top_{A,i}$ and an SVD of ${\bf Y}=\sum_{i=1}^k \sigma_{Y,i}{\bf u}_{Y,i}{\bf v}^\top_{Y,i}$. The constraint sets $\sigma_{Y,i} = \sqrt{T}$. Then the optimization problem reduces to:
\begin{align}
\min_{{\bf u}_{Y,1},\ldots,{\bf u}_{Y,k},{\bf v}_{Y,1},\ldots,{\bf v}_{Y,k}} -\frac 2{\sqrt{T}} \sum_{i=1}^k\sigma_{A,i}\sum_{j=1}^k {\bf u}_{A,i}^\top{\bf u}_{Y,j}{\bf v}_{A,i}^\top{\bf v}_{Y,j},\qquad {\rm s.t.}\quad {\bf u}_{Y,i}^\top{\bf u}_{Y,j}=\delta_{ij}, \quad {\bf v}_{Y,i}^\top{\bf v}_{Y,j}=\delta_{ij}.
\end{align}
Note that $\sum_{j=1}^k {\bf u}_{A,i}^\top{\bf u}_{Y,j}{\bf v}_{A,i}^\top{\bf v}_{Y,j} \leq 1$\footnote{Define $\alpha_j := {\bf u}_{A,i}^\top{\bf u}_{Y,j}$ and $\beta_j := {\bf v}_{A,i}^\top{\bf v}_{Y,j}$. Because  ${\bf u}_{Y,i}^\top{\bf u}_{Y,j}= {\bf v}_{Y,i}^\top{\bf v}_{Y,j}=\delta_{ij}$, it follows that $\sum_{i=1}^k\alpha^2_i =1$ and $\sum_{i=1}^k\beta^2_i\leq 1$. The sum in question is $\sum_{i=1}^k\alpha_i\beta_i$, which is an inner product of a unit vector and a vector with magnitude less than or equal to 1. Hence, the maximal inner product can be 1.} and therefore the cost is lower bounded by $-\frac 2{\sqrt{T}} \sum_{i=1}^k \sigma_{A,i} $. The lower bound is achieved when ${\bf u}_{A,i}={\bf u}_{Y,i}$ and ${\bf v}_{A,i}={\bf v}_{Y,i}$, with the optimal value of the objective $-\frac 2{\sqrt{T}} \sum_{i=1}^k \sigma_{A,i} = -\frac 2{\sqrt{T}}{\rm Tr}\left(\left({\bf A}{\bf A}^\top\right)^{1/2}\right)$.
 
\end{proof}

\begin{prop}\label{maxmin}Let ${\bf A}\in \mathbb{R}^{k\times T}$ with $k\leq T$. Then
\begin{align}
\max_{{\bf M}\in \mathbb{R}^{k\times k}} \min_{{\bf Y}\in \mathbb{R}^{k\times T}} -\frac 2T{\rm Tr}\left( {\bf A}^\top{\bf Y} \right)+{\rm Tr}\left( {\bf M}\left(\frac 1T{\bf Y}{\bf Y}^\top-{\bf I}\right) \right) =-\frac 2{T^{1/2}}{\rm Tr}\left(\left({\bf A}{\bf A}^\top\right)^{1/2}\right).
\end{align}
\end{prop}
\begin{proof} Note that we only need to consider the symmetric part of ${\bf M}$, because its antisymmetric component does not contribute to the cost. Below, we use ${\bf M}$ to mean its symmetric part. We will evaluate the value of the objective
\begin{align}\label{maxming}
-\frac 2T{\rm Tr}\left( {\bf A}^\top{\bf Y} \right)+{\rm Tr}\left( {\bf M}\left(\frac 1T{\bf Y}{\bf Y}^\top-{\bf I}\right) \right)
\end{align}
considering the following cases:
\begin{enumerate}
\item  ${\bf A} = {\bf 0}$. In this case the first term in Eq. \eqref{maxming} drops. Minimization of the second term with respect to ${\bf Y}$ gives $-\infty$ if ${\bf M}$ has a negative eigenvalue,  or a $0$ if ${\bf M}$ is positive semidefinite. Hence, the max-min objective is zero, and the proposition holds.
\item ${\bf A} \neq {\bf 0}$ and ${\bf A}$ is full-rank. 
\begin{enumerate}
\item ${\bf M}$ has at least one negative eigenvalue. Then, minimization of Eq. \eqref{maxming} with respect to ${\bf Y}$ gives $-\infty$.
\item ${\bf M}$ is positive semidefinite and has at least one zero eigenvalue. Then, minimization of Eq. \eqref{maxming} with respect to ${\bf Y}$ gives $-\infty$. To achieve this solution, one chooses all columns of ${\bf Y}$ to be one of the zero eigenvectors. The sign of the eigenvector is chosen such that ${\rm Tr}\left( {\bf A}^\top{\bf Y} \right)$ is positive. Multiplying ${\bf Y}$  by a positive scalar, one can reduce the objective indefinitely.
\item ${\bf M}$ is positive definite. Then, ${\bf Y}^* = {\bf M}^{-1}{\bf A}$ minimizes Eq. \eqref{maxming} with respect to ${\bf Y}$. Plugging this back to \eqref{maxming}, we get the objective 
\begin{align}\label{eqeq3}
 -\frac 1T{\rm Tr}\left( {\bf A}^\top{\bf M}^{-1}{\bf A}\right)-{\rm Tr}\left( {\bf M}\right).
\end{align}
The positive definite ${\bf M}$ that maximizes Eq. \eqref{eqeq3} can be found by setting its derivative to zero
\begin{align}
 {{\bf M}^*}^2 = \frac 1T {\bf A}{\bf A}^\top.
\end{align}
Plugging this back in Eq. \eqref{eqeq3}, one gets the objective
\begin{align}
 -\frac 2{\sqrt{T}}{\rm Tr}\left(\left({\bf A}{\bf A}^\top\right)^{1/2}\right),
\end{align}
which is maximal with respect to all possible ${\bf M}$. Therefore the proposition holds.

\end{enumerate}

\item ${\bf A} \neq {\bf 0}$ and ${\bf A}$ has rank $r<k$.
\begin{enumerate} 
\item ${\bf M}$ has at least one negative eigenvalue. Then, minimization of Eq. \eqref{maxming} with respect to ${\bf Y}$ gives $-\infty$, as before.
\item ${\bf M}$ is positive semidefinite and has at least one zero eigenvalue.
\begin{enumerate}
\item If at least one of the zero-eigenvectors of ${\bf M}$ is not a left zero-singular vector of ${\bf A}$, then, minimization of Eq. \eqref{maxming} with respect to ${\bf Y}$ gives $-\infty$. To achieve this solution, one chooses all columns of ${\bf Y}$ to be the zero-eigenvector of ${\bf M}$ that is not a left zero-singular vector of ${\bf A}$. The sign of the eigenvector is chosen such that ${\rm Tr}\left( {\bf A}^\top{\bf Y} \right)$ is positive. Multiplying ${\bf Y}$  by a positive scalar, one can reduce the objective indefinitely.
\item If all of the zero-eigenvectors of ${\bf M}$ are also left zero-singular vectors of ${\bf A}$, then Eq. \eqref{maxming} can be reformulated in the subspace spanned by top $r$ eigenvectors of ${\bf M}$. Suppose a SVD for ${\bf A} = \sum_{i=1}^r\sigma_{A,i}{\bf u}_{A,i}{\bf v}_{M,i}^\top$ with $\sigma_{A,1}\geq \sigma_{A,2}\geq \ldots \geq \sigma_{A,r}$. One can decompose ${\bf Y} = {\bf Y}^{A} + {\bf Y}^\perp$, where columns of ${\bf Y}^\perp$ are perdendicular to the space spanned by $\lbrace {\bf u}_{A,1},\ldots,{\bf u}_{A,r}\rbrace$. Then value of the objective Eq. \eqref{maxming} only depends on ${\bf Y}^{A}$. Defining new matrices $\tilde {\bf A}_{i,:} = {\bf u}_{A,i}^\top{\bf A}$, $\tilde {\bf Y}_{i,:} = {\bf u}_{A,i}^\top{\bf Y}^A$, $\tilde {\bf M}_{ij} = {\bf u}_{A,i}^\top{\bf M}{\bf u}_{A,j}$, where $i,j=1,\ldots,r$, we can rewrite Eq. \eqref{maxming} as  
\begin{align}\label{maxmingred}
-\frac 2T{\rm Tr}\left( \tilde{\bf A}^\top\tilde{\bf Y} \right)+{\rm Tr}\left( \tilde{\bf M}\left(\frac 1T\tilde{\bf Y}\tilde{\bf Y}^\top-{\bf I}\right) \right).
\end{align}
Now $\tilde{\bf A}$ is full-rank and $\tilde{\bf M}$ is positive definite. As in 2.(c), the objective which is maximal with respect to positive definite $\tilde{\bf M}$ matrices is
\begin{align}
-\frac 2{\sqrt{T}}{\rm Tr}\left(\left(\tilde{\bf A}\tilde{\bf A}^\top\right)^{1/2}\right) = -\frac 2{\sqrt{T}}{\rm Tr}\left(\left({\bf A}{\bf A}^\top\right)^{1/2}\right).
\end{align}

\end{enumerate}
\item ${\bf M}$ is positive definite. As in 2.(c), the objective which is maximal with respect to positive definite ${\bf M}$ matrices is
\begin{align}
 -\frac 2{\sqrt{T}}{\rm Tr}\left(\left({\bf A}{\bf A}^\top\right)^{1/2}\right).
\end{align}
This is also maximal with respect to all possible ${\bf M}$. Therefore the proposition holds.

\end{enumerate}
\end{enumerate}

Collectively, these arguments prove Eq. \eqref{maxming}.
\end{proof}

Propositions \eqref{minmax} and \eqref{maxmin} imply the strong min-max property for the PSW cost.

\begin{prop}\label{minmaxPSW}
The strong min-max property for the PSW cost:
\begin{align}
&\min_{{\bf Y}\in \mathbb{R}^{k\times T}}\max_{{\bf M}\in \mathbb{R}^{k\times k}} -\frac 2T{\rm Tr}\left( {\bf X}^\top{\bf W}^\top{\bf Y} \right)+{\rm Tr}\left( {\bf M}\left(\frac 1T{\bf Y}{\bf Y}^\top-{\bf I}\right) \right) \nonumber \\
&\quad= \max_{{\bf M}\in \mathbb{R}^{k\times k}}\min_{{\bf Y}\in \mathbb{R}^{k\times T}} -\frac 2T{\rm Tr}\left( {\bf X}^\top{\bf W}^\top{\bf Y} \right)+{\rm Tr}\left( {\bf M}\left(\frac 1T{\bf Y}{\bf Y}^\top-{\bf I}\right) \right) \nonumber \\
&\quad=-\frac 2{T^{1/2}}{\rm Tr}\left(\left({\bf W}{\bf X}{\bf X}^\top{\bf W}^\top\right)^{1/2}\right) . 
\end{align}
\end{prop}

\pagebreak

\section{Proof of Theorem \ref{mainLSPSW}}\label{proofmainLSPSW}

Here we prove Theorem \ref{mainLSPSW}. 

\subsection{Proof of item 1}

Item 1 directly follows from the fixed point equations of the dynamical system \eqref{gdPSW}, ( $\bar{}$ for fixed point).
\begin{align}\label{fpPSW}
\bar {\bf W} &=  \bar {\bf Y} {\bf X}^\top =  \bar {\bf F}{\bf C}  \nonumber \\
{\bf I} & = \bar {\bf Y}{\bar {\bf Y}}^\top= \bar{\bf F}{\bf C}\bar{\bf F}^\top.
\end{align}

\subsection{Proof of item 2}

We will prove item 2, making use of the normalized neural filters:
\begin{align}\label{Rdef}
{\bf R} := {\bf F}{\bf C}^{1/2},
\end{align}
where the input covariance matrix ${\bf C}$ is defined as in Eq. \eqref{Cdef}.	At the fixed point, the normalized neural filters are orthonormal: 
\begin{align}\label{Rorth}
\bar{\bf R}\bar{\bf R}^\top =  \bar {\bf F}{\bf C}\bar{\bf F}^\top = \bar {\bf Y}{\bar {\bf Y}}^\top ={\bf I}.
\end{align}
Normalized filters commute with the covariance matrix:
	\begin{align}\label{Rcom}
	\bar{\bf R}^\top\bar{\bf R} {\bf C}= {\bf C} \bar{\bf R}^\top\bar{\bf R}.
	\end{align}

\begin{proof}
	\begin{align}
	\bar{\bf R}^\top\bar{\bf R} {\bf C} &= {\bf C}^{1/2}\bar{\bf F}^\top\bar {\bf F}{\bf C}^{3/2} =  {\bf C}^{1/2}\bar{\bf F}^\top\bar {\bf W}{\bf C}^{1/2} =  {\bf C}^{1/2}\bar{\bf F}^\top\bar {\bf M}\bar {\bf F}{\bf C}^{1/2} \nonumber \\
	& = {\bf C}^{1/2}\bar{\bf W}^\top\bar {\bf F}{\bf C}^{1/2} = {\bf C}{\bf C}^{1/2}\bar{\bf F}^\top\bar {\bf F}{\bf C}^{1/2} = {\bf C} \bar{\bf R}^\top\bar{\bf R}. 
	\end{align}
\end{proof}
Therefore, as argued in Appendix \ref{1i2}, rows of ${\bf R}$ span a subspace spanned by some $k$ eigenvectors of ${\bf C}$. If ${\bf C}$ is invertible, rowspace of ${\bf F}$ is the same as ${\bf R}$ (follows from Eq. \eqref{Rdef}) and item 2 follows.

\subsection{Proof of item 3}

\subsubsection*{Preliminaries}

In order to perform a linear stability analysis, we linearize the system of equations \eqref{gdPSW} around the fixed point. The evolution of ${\bf W}$ and ${\bf M}$ perturbations follow from linearization of \eqref{gdPSW}:
\begin{align}\label{lgdPSW}
\tau\frac{d\delta{\bf M}}{dt} &=  {\bf \delta R} \bar{\bf R}^\top +   \bar {\bf R} {\bf \delta R}^\top, \nonumber \\
\frac{d\delta{\bf W}}{dt} &= 2 {\bf \delta R}{\bf C}^{1/2}  -2{\bf \delta W}.
\end{align}

Even though Eq. \eqref{gdPSW} depends on ${\bf W}$ and ${\bf M}$, we will find it convenient to change variables and work with ${\bf R}$, as defined in Eq. \eqref{Rdef}, and ${\bf M}$ instead. Since ${\bf R}$, ${\bf W}$ and ${\bf M}$ are interdependent, we express the perturbations of ${\bf R}$ in terms of ${\bf W}$ and ${\bf M}$ perturbations:
\begin{align}
{\bf \delta R} &= {\bf \delta M}^{-1}\bar{\bf W}{\bf C}^{1/2} + \bar {\bf M}^{-1}{\bf \delta W} {\bf C}^{1/2} = -\bar{\bf M} ^{-1}{\delta \bf M} \bar{\bf R} + \bar {\bf M}^{-1}{\bf \delta W}{\bf C}^{1/2},
\end{align}
which implies that 
\begin{align}
\frac{d\delta{\bf R}}{dt} &= - \bar {\bf M}^{-1}\frac{d\,{\bf \delta M}}{dt} \bar{\bf R} + \bar {\bf M}^{-1}\frac{d\,{\bf \delta W}}{dt} {\bf C}^{1/2}.
\end{align}

Plugging these in and eliminating $\delta{\bf W}$, we arrive at a linearized equation for $\delta {\bf R}$:
\begin{align}\label{evolR}
\frac{d\delta{\bf R}}{dt} = -\frac 1{\tau} \bar {\bf M}^{-1}\left( {\bf \delta R} \bar{\bf R}^\top +   \bar {\bf R} {\bf \delta R}^\top+2\tau{\bf \delta M}\right)\bar{\bf R} + 2\bar {\bf M}^{-1} {\bf \delta R}{\bf C} -2{\bf \delta R}.
\end{align}

To asses the stability of ${\bf \delta R}$, we expand it as in Appendix \ref{1i3}:
\begin{align}\label{expR}
{\bf \delta R} = {\bf \delta A}\bar {\bf R} + {\bf \delta S}\bar {\bf R} + {\bf \delta B}\bar{\bf G}
\end{align}
where $\delta{\bf A}$ is an $k\times k$ skew-symmetric matrix, $\delta{\bf S}$ is an $k\times k$ symmetric matrix and $\delta{\bf B}$ is an $k\times(n-k)$ matrix. $\bar{\bf G}$ is an $(n-k) \times n$ matrix with orthonormal rows. These rows are chosen to be orthogonal to the rows of $\bar{\bf R}$.  As before, skew-symmetric ${\bf \delta A}$ corresponds to rotations of filters within the normalized filter space, symmetric ${\bf \delta S}$ keeps the normalized filter space invariant but destroys orthonormality and $\delta {\bf B}$ is a perturbation that takes the normalized neural filters outside of the filter space. 

Let ${\bf v}_{1,\ldots,n}$ be the eigenvectors ${\bf C}$ and $\sigma_{1,\ldots,n}$ be the corresponding eigenvalues. We label them such that  $\bar{\bf R}$ spans the same space as the space spanned by the first $k$ eigenvectors. We choose rows of $\bar{\bf G}$ to be the remaining eigenvectors, i.e. $\bar{\bf G}^\top:=[{\bf v}_{k+1},\ldots,{\bf v}_{n}]$.

\subsubsection*{Proof}

Proof of item 3 of Theorem \ref{mainLSPSW} follows from studying the stability of ${\bf \delta B}$ component. Multiplying Eq. \eqref{evolR} on the right by $\bar{\bf G}^\top$, we arrive at a decoupled evolution equation:
\begin{align}
\frac {d \delta{B}_{i}^j}{dt} = \sum_m P^j_{im}\delta B_{m}^j, \qquad  P^j_{im}:= 2\left(\bar M^{-1}_{im}\sigma_{j+k} -\delta_{im}\right),
\end{align}
where for convenience we change our notation to $\delta B_{kj}=\delta B_{k}^j$. 

Eq.s \eqref{fpPSW} and \eqref{Rorth} imply $\bar{\bf M}^2 = \bar{\bf W}{\bf  C}\bar{\bf W}^\top= \bar{\bf R}{\bf C}^2\bar{\bf R}^\top$
and hence:
\begin{align}
\bar{\bf M} = \bar{\bf R}{\bf C}\bar{\bf R}^\top.
\end{align}
Taking into account Eq.s \eqref{Rorth} and \eqref{Rcom}, the case at hand reduces to the proof presented in Appendix \ref{1i3}: stable solutions are those for which
\begin{align}
\left\lbrace \sigma_{1}, \ldots,\sigma_{k}\right\rbrace > \left\lbrace \sigma_{k+1}, \ldots,\sigma_{n}\right\rbrace.
\end{align}
This proves that if at the fixed point, normalized neural filters span a subspace other than the principal subspace, the fixed point is linearly unstable. Since the span of normalized neural filters is that of the neural filters, item 3 follows.

\subsection{Proof of item 4}

Proof of item 4 follows from the linear stabilities of $\delta {\bf A}$ and $\delta {\bf S}$. Multiplying Eq. \eqref{evolR} on the right by $\bar{\bf R}^\top$, and separating the resulting equation in to into its symmetric and anti-symmetric parts, we arrive at:
\begin{align}
\frac{d {\bf \delta A}} {dt}  &= -\frac 1{\tau}\left(\bar{\bf M}^{-1}{\bf \delta S}-{\bf \delta S}\bar{\bf M}^{-1}\right)-\bar{\bf M}^{-1}{\bf \delta M}+{\bf \delta M}\bar{\bf M}^{-1}- 2{\bf \delta A} \nonumber \\
&\qquad\qquad + \bar{\bf M}^{-1}{\bf \delta A}\bar{\bf M}+\bar{\bf M}{\bf \delta A}\bar{\bf M}^{-1} + \bar{\bf M}^{-1}{\bf \delta S}\bar{\bf M}-\bar{\bf M}{\bf \delta S}\bar{\bf M}^{-1}, \nonumber \\
\frac{d {\bf \delta S}} {dt}  &= -\frac 1{\tau}\left(\bar{\bf M}^{-1}{\bf \delta S}+{\bf \delta S}\bar{\bf M}^{-1}\right)-\bar{\bf M}^{-1}{\bf \delta M}-{\bf \delta M}\bar{\bf M}^{-1}- 2{\bf \delta S} \nonumber \\
&\qquad\qquad + \bar{\bf M}^{-1}{\bf \delta A}\bar{\bf M}-\bar{\bf M}{\bf \delta A}\bar{\bf M}^{-1} + \bar{\bf M}^{-1}{\bf \delta S}\bar{\bf M}+\bar{\bf M}{\bf \delta S}\bar{\bf M}^{-1}
\end{align}
To obtain a closed set of equations, we complement these equations with ${\bf \delta M}$ evolution, which we obtain by plugging the expansion \eqref{expR} into Eq. \eqref{lgdPSW}:
\begin{align}
\tau\frac{d {\bf \delta M}} {dt}  =2{\bf \delta S}
\end{align}
We only consider symmetric ${\bf \delta M}$ below, since our algorithm preserves the symmetry of ${\bf M}$ in runtime.

We now change to a basis where $\bar{\bf M}$ is diagonal. $\bar{\bf M}$ is symmetric and has an orthonormal set of eigenvectors. Its eigenvalues are the principal eigenvalues $\lbrace \sigma_1, \ldots, \sigma_k\rbrace$ (from Appendix \ref{1i3}). Let ${\bf U}$ be the matrix that contains the eigenvectors of $\bar{\bf M}$ in its columns. Define
\begin{align}
{\bf \delta A}^U &:= {\bf U}^\top{\bf \delta A}{\bf U}, \nonumber \\
{\bf \delta S}^U &:= {\bf U}^\top{\bf \delta S}{\bf U}, \nonumber \\
{\bf \delta M}^U &:= {\bf U}^\top{\bf \delta M}{\bf U}.
\end{align}
In this new basis, the linearized equations, in component form, become:
\begin{align}
\frac{d} {dt}  \left[\begin{array}{cc}  {\delta M}^U_{ij} \\ {\delta A}^U_{ij} \\ {\delta S}^U_{ij}\end{array}\right] = {\bf H}^{ij}\left[\begin{array}{cc}  {\delta M}^U_{ij} \\ {\delta A}^U_{ij} \\ {\delta S}^U_{ij}\end{array}\right],
\end{align}
where
\begin{align}
{\bf H}^{ij} := \left[\begin{array}{ccc} 
 0 & 0 & \frac 2{\tau} \\
 \frac{1}{\sigma_j}-\frac{1}{\sigma_i} & - 2+ \frac{\sigma_j}{\sigma_i}+\frac{\sigma_i}{\sigma_j} & -\frac 1{\tau}\left(\frac 1{\sigma_i}-\frac 1{\sigma_j}\right)+\frac{\sigma_j}{\sigma_i}-\frac{\sigma_i}{\sigma_j} \\
 -\frac{1}{\sigma_j}-\frac{1}{\sigma_i}& \frac{\sigma_j}{\sigma_i}-\frac{\sigma_i}{\sigma_j}& -\frac 1{\tau}\left(\frac 1{\sigma_i}+\frac 1{\sigma_j}\right)+\frac{\sigma_j}{\sigma_i}+\frac{\sigma_i}{\sigma_j}-2
\end{array}\right]
\end{align}

Linear stability is governed by the three eigenvalues of ${\bf H}^{ij}$. One of the eigenvalues is $0$, due to the existence of the rotational symmetry in the problem. The corresponding eigenvector is $\left [ \sigma_{j}-{\sigma_i}, 1, 0 \right]$. Note that the third element of the eigenvector is zero, showing that the orthogonality of the normalized neural filters are not spoiled even in this mode.

For stability of the principal subspace, the other two eigenvalues must be negative, which means their sum should be negative, and their multiplication should be positive. It is easy to show that both the negativity of the summation and the positivity of the multiplication holds if and only if for all $(i,j)$ pairs with $i\neq j$:
\begin{align}
 \tau < \frac{\sigma_i+\sigma_j}{2\left(\sigma_i-\sigma_j\right)^2}.
\end{align}

Hence we have showed that linear perturbations of fixed point weights decay to a configuration in which normalized neural filters are rotations of the original normalized neural filters within the subspace. It follows from Eq. \eqref{Rdef}, that the same holds for neural filters.

\pagebreak

\section{Autapse-free similarity matching network with asymmetric lateral connectivity}\label{coord}

Here, we derive an alternative neural network algorithm for PSP, which does not feature autaptic connections and has asymmetric lateral connections. To this end, we replace the gradient descent neural dynamics defined by Eq. \eqref{gdMDSon} by a coordinate descent dynamics. 

In the coordinate descent approach, at every step, one finds the optimal value of one component of ${\bf y}_t$, while keeping the rest fixed. By taking the derivative of the cost $-4{\bf x}_t^\top{\bf W}{\bf y}_t + 2{\bf y}_t^\top{\bf M}{\bf y}_t$ with respect to $y_{t,i}$ and setting it to zero we find:
\begin{align}\label{cd}
y_{t,i} = \sum_{j=1}\frac{W_{t,ij}}{M_{t,ii}}x_{t,j} - \sum_{j\neq i}\frac{M_{t,ij}}{M_{t,ii}}y_{t,j}.
\end{align}
The components can be cycled through in any order until the iteration converges to a fixed point. The iteration is guaranteed to converge under very mild assumptions: diagonals of ${\bf M}$ have to be positive \citep{luo1991convergence}, which is satisfied if ${\bf M}$ is initialized that way, see Eq. \eqref{wmSM}. Finally, Eq. \eqref{cd} can be interpreted as a Gauss-Seidel iteration and generalizations to other iterative schemes are possible, see \citep{pehlevan2015MDS}. 

The coordinate descent iteration, Eq. \eqref{cd}, can be interpreted as the dynamics of an asynchronous autapse-free neural network, Fig. \ref{Fig1}B, where synaptic weights are:
\begin{align}\label{tildeWM}
\tilde W_{t,ij} &= \frac{W_{t,ij}}{M_{t,ii}}, \qquad
\tilde M_{t,ij} = \frac{M_{t,,j}}{M_{t,ii}}, \qquad \tilde M_{t,ii} = 0.
\end{align}
With this definition, the lateral weights are now asymmetric because $M_{t,ii} \neq M_{t,jj}$ if $i\neq j$. 

We can derive updates for these synaptic weights from the updates for ${\bf W}_t$ and ${\bf M}_t$, Eq. \eqref{wmSM}. By defining another scalar state variable for each $i$th neuron  $\tilde {D}_{t,i} := \tau M_{t,ii}/\eta_{t-1} $, we arrive at\footnote{These update rules can be derived as follows. Start by the definition of the synaptic weights, Eq. \eqref{tildeWM}: $M_{t+1,ii}\tilde{M}_{t+1,ij} = {M}_{t+1,ij}$. By the gradient-descent update Eq. \eqref{wmSM}, ${M}_{t+1,ij}= \left(1-\frac {\eta_{t}}{\tau}\right) {M}_{t,ij} + \frac{\eta_{t}}{\tau} y_{t,i}y_{t,j} = \left(1-\frac{\eta_{t}}{\tau}\right) \tilde {M}_{t,ij}M_{t,ii}  + \frac{\eta_{t}}{\tau} y_{t,i}y_{t,j} $, where in the second equality we again  used Eq. \eqref{tildeWM}. But note that $(1-\frac{\eta_{t}}{\tau})M_{t,ii} = M_{t+1,ii} - \frac{\eta_{t}}{\tau} y_{t,i}^2$, from Eq. \eqref{wmSM}. Combining all of these, $\tilde{M}_{t+1,ij} = \tilde {M}_{t,ij} + \frac {\eta_{t}}{\tau M_{t+1,ii}} \left(y_{t,i}x_{t,j}-y_{t,i}^2\tilde{M}_{t,ij}\right) $. Similar derivation can be given for feedforward updates.}:

\begin{align}
\tilde {D}_{t+1,i} &= \frac{\eta_{t-1}}{\eta_{t}}\left(1-\frac{\eta_{t}}{\tau}\right) \tilde {D}_{t,i} + y_{t,i}^2, \nonumber \\ 
\tilde W_{t+1,ij} &= \left(\frac{1-2\eta_{t}}{1-\eta_{t}/\tau}\right)\tilde W_{t,ij} + \frac{1}{\tilde{D}_{t+1,i}}\left(2\tau y_{t,i}x_{t,j}-\left(\frac{1-2\eta_{t}}{1-\eta_{t}/\tau}\right)y_{t,i}^2\tilde{W}_{t,ij}\right), \nonumber \\
\tilde M_{t+1,i,j\neq i} &= \tilde M_{t,ij} + \frac{1}{\tilde{D}_{t+1,i}}\left(y_{t,i}y_{t,j}-y_{t,i}^2\tilde{M}_{t,ij}\right). \nonumber \\
\tilde M_{t+1,ii} &= 0,
\end{align}
Here, in addition to synaptic weights, the neurons need to keep track of a postsynaptic activity depended variable $\tilde D_{t,i}$ and the gradient descent-ascent learning rate parameters $\eta_{t}$,  $\eta_{t-1}$ and $\tau$.  The updates are local. 

For the special case of $\tau=1/2$ and $\eta_t = \eta/2$, these plasticity rules simplify to,
\begin{align}
\tilde {D}_{t+1,i} &= (1-\eta) \tilde {D}_{t,i} + y_{t,i}^2, \nonumber \\ 
\tilde W_{t+1,ij} &= \tilde W_{t,ij} + \frac{1}{\tilde{D}_{t+1,i}}\left(y_{t,i}x_{t,j}-y_{t,i}^2\tilde{W}_{t,ij}\right) \nonumber \\
\tilde M_{t+1,i,j\neq i} &= \tilde M_{t,ij} + \frac{1}{\tilde{D}_{t+1,i}}\left(y_{t,i}y_{t,j}-y_{t,i}^2\tilde{M}_{t,ij}\right). \nonumber \\
\tilde M_{t+1,ii} &= 0,
\end{align}
which is precisely the neural online similarity matching algorithm we previously gave in \citep{pehlevan2015MDS}. Both feedforward and lateral updates have the same form as a single-neuron Oja's rule \citep{oja1982simplified}. 

Note that the algorithm derived above is essentially the same as the one in the main text: given the same initial conditions and the same inputs, ${\bf x}_t$, they will produce the same outputs,  ${\bf y}_t$. The only difference is a rearrangement of synaptic weights in the neural network implementation.

\pagebreak

\section{Autapse-free constrained similarity matching network with asymmetric lateral connectivity}\label{coordPSW}

Following similar steps to Appendix \ref{coord}, we derive an autapse-free PSW neural algorithm with asymmetric lateral connections. We replace the gradient descent neural dynamics defined by Eq. \eqref{gdCSMon} by a coordinate descent dynamics, where at every step, one finds the optimal value of one component of ${\bf y}_t$, while keeping the rest fixed:
\begin{align}\label{cdCSM}
y_{t,i} = \sum_{j=1}\frac{W_{t,ij}}{M_{t,ii}}x_{t,j} - \sum_{j\neq i}\frac{M_{t,ij}}{M_{t,ii}}y_{t,j}.
\end{align}
The components can be cycled through in any order until the iteration converges to a fixed point. 

The coordinate descent iteration, Eq. \eqref{cdCSM}, can be interpreted as the dynamics of an asynchronous autapse-free neural network, Fig. \ref{Fig1}B, with synaptic weights:
\begin{align}\label{tildeWMCSM}
\tilde W_{t,ij} &= \frac{W_{t,ij}}{M_{t,ii}}, \qquad
\tilde M_{t,ij} = \frac{M_{t,,j}}{M_{t,ii}}, \qquad \tilde M_{t,ii} = 0.
\end{align}
As in Appendix \ref{coord}, the new lateral weights are asymmetric.

Updates for these synaptic weights can be derived from the updates for ${\bf W}_t$ and ${\bf M}_t$, Eq. \eqref{wmCSM}.  Defining another scalar state variable for each $i$th neuron  $\tilde {D}_{t,i} := \tau M_{t,ii}/\eta_{t-1} $, we arrive at

\begin{align}
\tilde {D}_{t+1,i} &= \frac{\eta_{t-1}}{\eta_{t}}\left(1-\frac{\eta_{t}}{\tau}\right) \tilde {D}_{t,i} + y_{t,i}^2-1, \nonumber \\ 
\tilde W_{t+1,ij} &= \left(1-2\eta_{t}\right)\tilde W_{t,ij} + \frac{1}{\tilde{D}_{t+1,i}}\left(2\tau y_{t,i}x_{t,j}-\left(1-2\eta_{t}\right)\left(y_{t,i}^2-1\right)\tilde{W}_{t,ij}\right), \nonumber \\
\tilde M_{t+1,i,j\neq i} &= \tilde M_{t,ij} + \frac{1}{\tilde{D}_{t+1,i}}\left(y_{t,i}y_{t,j}-\left(y_{t,i}^2-1\right)\tilde{M}_{t,ij}\right), \nonumber \\
\tilde M_{t+1,ii} &= 0.
\end{align}
As in Appendix \ref{coord}, in addition to synaptic weights, the neurons need to keep track of a postsynaptic activity depended variable $\tilde D_{t,i}$ and gradient descent-ascent learning rate parameters $\eta_{W,t}$, $\eta_{M,t}$ and $\eta_{M,t-1}$. 

For the special case of $\eta_t =\eta/2$ and $\tau = 1/2$, these plasticity rules simplify to,
\begin{align}
\tilde {D}_{t+1,i} &= (1-\eta) \tilde {D}_{t,i} + y_{t,i}^2, \nonumber \\ 
\tilde W_{t+1,ij} &= (1-\eta)\tilde W_{t,ij} + \frac{1}{\tilde{D}_{t+1,i}}\left(y_{t,i}x_{t,j}-(1-\eta)\left(y_{t,i}^2-1\right)\tilde{W}_{t,ij}\right) \nonumber \\
\tilde M_{t+1,i,j\neq i} &= \tilde M_{t,ij} + \frac{1}{\tilde{D}_{t+1,i}}\left(y_{t,i}y_{t,j}-\left(y_{t,i}^2-1\right)\tilde{M}_{t,ij}\right), \nonumber \\
\tilde M_{t+1,ii} &= 0.
\end{align}

\pagebreak


\end{document}